\documentclass[12pt]{article}
\usepackage{amsmath}
\usepackage{amssymb, amscd, amsthm,amsfonts}
\usepackage[all]{xy}
\usepackage[dvips]{graphicx}
\usepackage{verbatim}
\newtheorem{theo}[]{{\emph{Theorem}}}

\newtheorem{lemma}[]{{\emph{Lemma}}}

\theoremstyle{remark}
\newtheorem*{remark}{\textbf{Remark}}

\theoremstyle{definition}

\newcommand{\ra}{\rightarrow}

\newcommand{\calC}{\mathcal{C}}

\newcommand{\calF}{\mathcal{F}} 

\newcommand{\bF}{\mathbb{F}}

\newcommand{\cC}{\mathcal{C}}

\newcommand{\al}{\alpha}
\newcommand{\be}{\beta}
\newcommand{\ga}{\gamma}
\newcommand{\veps}{\varepsilon}

\newcommand{\om}{\omega}

\DeclareMathOperator{\Tra}{\mathrm Tr}

\begin{document}

\title{Cyclic Codes and Sequences from a Class of Dembowski-Ostrom Functions} \maketitle

\begin{center}
$\mathrm{Jinquan\;\; Luo\qquad\quad San\;\; Ling \quad\qquad
and\qquad Chaoping\;\; Xing}$ \footnotetext{J.Luo is with the School
of Mathematics, Yangzhou University, Jiangsu Province, 225009, China
and with the Division of Mathematics, School of Physics and
Mathematical Sciences, Nanyang Technological University, Singapore.
\par S.Ling and C.Xing are with the Division of Mathematics, School
of Physics and Mathematical Sciences, Nanyang Technological
University, Singapore.
\par \quad E-mail addresses: jqluo@ntu.edu.sg, lingsan@ntu.edu.sg, xingcp@ntu.edu.sg.}
\end{center}
\newpage
 \textbf{Abstract} \par Let $q=p^n$ with $p$ be an odd
 prime.
Let $0\leq k\leq n-1$ and $k\neq n/2$. In this paper we determine
the value distribution of following exponential(character) sums
\[\sum\limits_{x\in
\bF_q}\zeta_p^{\Tra_1^n(\alpha x^{p^{3k}+1}+\beta
x^{p^k+1})}\quad(\alpha\in \bF_{p^m},\beta\in \bF_{q})\]
 and
\[\sum\limits_{x\in
\bF_q}\zeta_p^{\Tra_1^n(\alpha x^{p^{3k}+1}+\beta x^{p^k+1}+\ga
x)}\quad(\alpha\in \bF_{p^m},\beta,\ga\in \bF_{q})\]

 where $\Tra_1^n: \bF_q\ra \bF_p$ and $\Tra_1^m: \bF_{p^m}\ra\bF_p$ are the canonical trace
mappings and $\zeta_p=e^{\frac{2\pi i}{p}}$ is a primitive $p$-th
root of unity. As applications:
 \begin{itemize}
    \item[(1).]  We determine the weight distribution
of the cyclic codes $\cC_1$ and $\cC_2$ over $\bF_{p^t}$ with
parity-check polynomials $h_2(x)h_3(x)$ and $h_1(x)h_2(x)h_3(x)$
respectively where $t$ is a divisor of $d=\gcd(n,k)$, and $h_1(x)$,
$h_2(x)$ and $h_3(x)$ are the minimal polynomials of $\pi^{-1}$,
$\pi^{-(p^k+1)}$ and $\pi^{-(p^{3k}+1)}$ over $\bF_{p^t}$
respectively for a primitive element $\pi$ of $\bF_q$.
    \item[(2).]We determine the correlation distribution among
    a family of m-sequences. \end{itemize}

\newpage
\section{Introduction}

\quad Basic results on finite fields could be found in \cite{Lid Nie}. These notations are fixed throughout this paper except for specific statements.
\begin{itemize}
  \item Let $p$ an odd prime, $p^*=(-1)^{\frac{p-1}{2}}p$, $q=p^n$ and $\bF_p$, $\bF_q$ be the finite fields
  of order $p$, $q$ respectively. Let $\pi$ be a primitive element of $\bF_q$.
  \item Let $\Tra_i^j:\bF_{p^i}\ra\bF_{p^j}$ be the trace mapping, $\zeta_p=\mathrm{exp}(2\pi
\sqrt{-1}/p)$ be a $p$-th root of unity and
$\chi(x)=\zeta_p^{\Tra_1^n(x)}$ be the canonical additive character on
$\bF_q$.
  \item Let $k$ be a positive integer, $1\leq k\leq n-1$ and  $k\notin \{\frac{n}{4},\frac{n}{2},\frac{3n}{4}\}$.
  Let $d=\gcd(n,k)$, $q_0=p^d$, $q_0^*=(-1)^{\frac{q_0-1}{2}}q_0$ and $s=n/d$. Let $t$ be a divisor of $d$ and $n_0=n/t$.
  \item Let $m=n/2$ (if $n$ is even) and $\mu=(-1)^{m/d}$.
\end{itemize}

\quad Let $C$ be an $[l,k,d]$ linear code and $A_i$ be the number of
codewords in $\cC$ with Hamming weight $i$. The weight distribution
$\{A_i\}_{i=0}^{l}$ is an important research object for theoretical
and application interests(see Fitzgerald and Yucas \cite{Fit Yuc},
McEliece \cite{McEl}, McEliece and Rumsey \cite{McE Rum},  van der
Vlugt \cite{Vand}, Wolfmann \cite{Wolf} and the references therein).

For a cyclic code, the Hamming weight of each codeword can be
expressed by certain combination of general exponential(character)
sums (see Feng and Luo \cite{Fen Luo}, \cite{Fen Luo2}, Luo and Feng
\cite{Luo Fen}, \cite{Luo Fen2}, Luo, Tang and Wang \cite{Luo Tan Wan}, Luo \cite{Luo},
 van der Vlugt \cite{Vand2}, Yuan,
Carlet and Ding \cite{Yua Car}, Zeng, Hu, Jia, Yue and Cao \cite{Zen Hu Jia Yue Cao}, Zeng and Li \cite{Zen Li}). More exactly speaking, let $t\mid n$, $\cC$ be the cyclic code over $\bF_{p^t}$ with
length $l=q-1$ and parity-check polynomial,
\[h(x)=h_1(x)\cdots h_u(x)\quad (u\geq 2)\]
where $h_i(x)$ $(1\leq i\leq e)$ are distinct irreducible
polynomials in $\bF_{p^t}[x]$ with degree $e_i$ $(1\leq i\leq u)$,
then $\mathrm{dim}_{\bF_{p^t}}\cC=\sum\limits_{i=1}^{u}e_i$. Let
 $\pi^{-s_i}$ be a zero
of $h_i(x)$, $1\leq s_i\leq q-2$ $(1\leq i\leq u).$ Then the
codewords in $\cC$ can be expressed by
\[c(\alpha_1,\cdots,\alpha_u)=(c_0,c_1,\cdots,c_{l-1})\quad (\alpha_1,\cdots,\alpha_u\in \bF_q)\]
where
$c_i=\sum\limits_{\lambda=1}^{u}\Tra^n_{t}(\alpha_{\lambda}\pi^{is_{\lambda}})$
$(0\leq i\leq n-1)$. Therefore the Hamming
weight of the codeword $c=c(\alpha_1,\cdots,\alpha_u)$ is
{\setlength\arraycolsep{2pt}
\begin{eqnarray} \label{Wei}
w_H\left(c\right)&=& \#\left\{i\left|0\leq i\leq l-1,c_i\neq
0\right.\right\}
\nonumber\\[1mm]
&=& l-\#\left\{i\left|0\leq i\leq l-1,c_i=0\right.\right\}
\nonumber\\[1mm]
&=& l-\frac{1}{p^t}\,\sum\limits_{i=0}^{l-1}\sum\limits_{a\in
\bF_{p^t}}\zeta_p^{\Tra_1^{t}\left(a\cdot\Tra_{t}^n\left(\sum\limits_{\lambda=1}^{u}\alpha_{\lambda}\pi^{is_{\lambda}}\right)\right)}
\nonumber\\[1mm]
&=&l-\frac{l}{p^t}-\frac{1}{p^t}\,\sum\limits_{a\in
\bF_{p^t}^*}\sum\limits_{x\in \bF_q^*}\zeta_p^{\Tra_1^n(af(x))}
\nonumber
\\[1mm]
&=&l-\frac{l}{p^t}+\frac{p^t-1}{p^t}-\frac{1}{p^t}\,\sum\limits_{a\in \bF_{p^t}^*}S(a\alpha_1,\cdots,a\alpha_u)\nonumber\\[1mm]
&=&p^{n-t}(p^t-1)-\frac{1}{p^t}\,\sum\limits_{a\in
\bF_{p^t}^*}S(a\alpha_1,\cdots,a\alpha_u)
\end{eqnarray}
} where
$f(x)=\alpha_1x^{s_1}+\alpha_2x^{s_2}+\cdots+\alpha_ux^{s_u}\in
\bF_{q}[x]$, $\bF_q^*=\bF_q\backslash\{0\}$,
$\bF_{p^t}^*=\bF_{p^t}\backslash\{0\}$,  and
\[S(\alpha_1,\cdots,\alpha_u)=\sum\limits_{x\in \bF_q}\zeta_p^{\Tra_1^n(\alpha_1x^{s_1}+\cdots+\alpha_ux^{s_u})}.\]
In this way, the weight distribution of cyclic code $\cC$ can be
derived from the explicit evaluating of the exponential sums
\[S(\alpha_1,\cdots,\alpha_u)\quad(\alpha_1,\cdots,\alpha_u\in \bF_q).\]

Let $h_1(x)$, $h_2(x)$ and
$h_{3}(x)$ be the minimal polynomials of $\pi^{-1},\pi^{-(p^k+1)}$
and $\pi^{-{(p^{3k}+1)}}$ over $\bF_{p^t}$ respectively. Then
$\mathrm{deg}\,h_i(x)=n_0\; \text{for}\; 1\leq i\leq 3.$

 Let
$\cC_1$ and $\cC_2$ be the cyclic codes over $\bF_{p^t}$ with length
$l=q-1$ and parity-check polynomials $h_2(x)h_3(x)$ and
$h_1(x)h_2(x)h_3(x)$ respectively. Then we know that the dimensions
of $\cC_1$ and $\cC_2$ over $\bF_{p^t}$ are $2n_0$ and $3n_0$
respectively.

A Dembowski-Ostrom function on $\bF_q$ is a $\bF_q$-linear combination of $x^{p^i+p^j}$ with $0\leq i\leq n-1$.
Let $f_{\al,\be}(x)=\alpha
x^{p^{3k}+1}+\beta x^{p^k+1}$
for $\al,\be\in \bF_q$.  Define the exponential sums
\begin{equation}\label{def T}
T(\al,\be)=\sum\limits_{x\in \bF_q}\zeta_p^{\Tra_1^n\left(f_{\al,\be}(x)\right)}
\end{equation}
and for $\ga\in \bF_q$,
\begin{equation}\label{def S}
S(\al,\be,\ga)=\sum\limits_{x\in \bF_q}\zeta_p^{\Tra_1^n\left(f_{\al,\be}(x)+\ga x\right)}.
\end{equation}
Then the weight distribution of $\cC_1$ and $\cC_2$ can be
completely determined if $T(\al,\be)$ and $S(\al,\be,\ga)$ are
explicitly evaluated.

Another application of $S(\al,\be,\ga)$ is to calculate the
correlation distribution of corresponding sequences. Let
$\mathcal{F}$ be a collection of $p$-ary m-sequences of period $q-1$
defined by

\[\mathcal{F}=\left\{\left\{a_i(t)\right\}_{i=0}^{q-2}|\,0\leq i\leq L-1 \right\}\]

The \emph{correlation function} of $a_i$ and $a_j$ for a shift
$\tau$ is defined by

\[M_{{i},{j}}(\tau)=\sum\limits_{\lambda=0}^{q-2}\zeta_p^{a_i({\lambda})-a_j({\lambda+\tau})}\hspace{2cm}(0\leq \tau\leq
q-2).\]

 In this paper, we will study the collection of sequences

 \begin{equation}\label{def F}
  \calF=\left\{a_{\al,\be}=\left\{a_{\al,\be}(\pi^{\lambda})\right\}_{\lambda=0}^{q-2}|\,\al, \be\in \bF_{q} \right\}
 \end{equation}
where
  $a_{\al,\be}(\pi^{\lambda})=\Tra_1^n(\al \pi^{\lambda(p^{3k}+1)}+\be
  \pi^{\lambda(p^k+1)}+\pi^{\lambda})$.

Then the correlation function between $a_{\al_1,\be_1}$ and
$a_{\al_2,\be_2}$ by a shift $\tau$ ($0\leq \tau\leq q-2$) is
\begin{equation}\label{cor fun}
\begin{array}{ll}
&M_{(\al_1,\be_1),(\al_2,\be_2)}(\tau)=\sum\limits_{\lambda=0}^{q-2}\zeta_p^{a_{\al_1,\be_1}({\lambda})-
a_{\al_2,\be_2}({\lambda+\tau})}\\[2mm]
&\qquad =\sum\limits_{\lambda=0}^{q-2}\zeta_p^{\Tra_1^n(\al_1
\pi^{\lambda(p^{3k}+1)}+\be_1
  \pi^{\lambda(p^k+1)}+\pi^{\lambda})-\Tra_1^n(\al_2 \pi^{(\lambda+\tau)(p^{3k}+1)}+\be
  \pi^{(\lambda+\tau)(p^k+1)}+\pi^{\lambda+\tau})}\\[2mm]
  &\qquad = S(\al',\be',\ga')-1
  \end{array}
\end{equation}
 where
 \begin{equation}\label{coe cor}
 \al'=\al_1-\al_2 \pi^{\tau(p^{3k}+1)},\quad
 \be'=\be_1-\be_2\pi^{\tau(p^k+1)},\quad \ga'=1-\pi^{\tau}.
 \end{equation}

Pairs of non-binary m-sequences with few-valued cross correlations have
been extensively studied for several decades, see Charpin \cite{Char}, Dobbertin, Helleseth, Kumar and Martinsen
 \cite{Dob Hel Kum Mar}, Gold \cite{Gold}, Helleseth \cite{Hell1}, \cite{Hell2},
Helleseth and Kumar \cite{Hel Kum}, Helleseth, Lahtonen and
Rosendahl \cite{Hel Lah Ros}, Kasami \cite{Kasa1},\cite{Kasa2},  Rosendahl
\cite{Rose1}, \cite{Rose2} and Trachtenberg \cite{Trac},
Xia and Zeng \cite{Xia Zen} and references therein.

In \cite{Zen Hu Jia Yue Cao}, the expoenential sums $T(\al,\be)$ and $S(\al,\be,\ga)$ for $n/d$ odd have been evaluated.
As an application, the weight distribution to the associated $p$-ary cyclic code $\cC_2$ is determined. Our paper
focuses on the case $n/d$ is even. Moreover, we will determine the weight distribution of $\cC_1$ and $\cC_2$
for $t\mid d$. Meanwhile, the correlation distribution of sequences in $\calF$ can also be calculated explicitly.

 This paper is presented as follows. In Section 2 we introduce
some preliminaries. In Section 3 we will study the value
distribution of $T(\al,\be)$ (that is, which value $T(\al,\be)$ takes  on and which frequency of each value) and
the weight distribution of $\cC_1$. In Section 3 we will determine
the value distribution of $S(\al,\be,\ga)$ , the correlation
distribution among the sequences in $\calF$, and then the weight
distribution of $\cC_2$. Most lengthy details are presented in
several appendixes. The main tools are quadratic form theory over
odd characteristic finite fields, some moment identities on
$T(\alpha,\beta)$ and a class of Artin-Schreier curves on finite
fields which we have employed in \cite{Luo Fen} and \cite{Luo Tan Wan}.
 We will focus our study on the odd prime characteristic case
and the binary case will be investigated in a following paper.

\section{Preliminaries}

\quad We follow the notations in Section 1.  The first machinery to
determine the values of exponential sums $T(\alpha,\beta)$ and
$S(\al,\be,\ga)$ defined in (\ref{def T}) and (\ref{def S}) is
quadratic form theory over $\bF_{q_0}$.

 Let $H$
be an $s\times s$ symmetric matrix over $\bF_{q_0}$ and
$r=\mathrm{rank}\,H$. Then there exists $M\in
\mathrm{GL}_s(\bF_{q_0})$ such that $H'=MHM^T$ is diagonal and
$H'=diag(a_1,\cdots,a_r,0,\cdots,0)$ where $a_i\in \bF_{q_0}^*$
($1\leq i\leq r$). Let $\Delta=a_1\cdots a_r$ (we assume $\Delta=1$
when $r=0$) and $\eta_0$ be the quadratic (multiplicative) character
of $\bF_{q_0}$. Then $\eta_0(\Delta)$ is an invariant of $H$ under
the conjugate action of $M\in \mathrm{GL}_s(\bF_{q_0})$.

For the quadratic form
\begin{equation}\label{qua for}
F:\bF_{q_0}^s\ra \bF_{q_0},\quad F(x)=XHX^T\quad
(X=(x_1,\cdots,x_s)\in \bF_{q_0}^s),
\end{equation}
 we have the following result(see \cite{Luo Fen}, Lemma 1).
\begin{lemma}\label{qua}
(i). For the quadratic form $F=XHX^T$ defined in (\ref{qua for}), we
have
\[
\sum\limits_{X\in\bF_{q_0}^s}\zeta_p^{\Tra_1^{d}(F(X))} =\left\{
\begin{array}{ll}
\eta_0(\Delta)q_0^{s-r/2} & \ \ \hbox{if} \ q_0\equiv
1\;(\mathrm{mod}\;
4),\\[2mm]
i^r\eta_0(\Delta){q_0}^{s-r/2}  & \ \ \hbox{if} \ {q_0}\equiv 3\;(\mathrm{mod}\; 4).\\
\end{array}
\right.
\]
(ii). For $A=(a_1,\cdots,a_s)\in \bF_{q_0}^s$, if $2YH+A=0$ has
solution $Y=B\in \bF_{q_0}^s$,

then
$\sum\limits_{X\in\bF_{q_0}^s}\zeta_p^{\Tra_1^{d}(F(X)+AX^T)}=\zeta_p^c\sum\limits_{X\in
\bF_{q_0}^s}\zeta_p^{\Tra_1^{d}\left({F(X)}\right)}$ where
$c=-\Tra_1^{d}\left(BHB^T\right)=\frac{1}{2}\Tra_1^{d}\left(AB^T\right)\in
\bF_p$.

Otherwise
$\sum\limits_{X\in\bF_p^m}\zeta_p^{\Tra_1^{d}(F(X)+AX^T)}=0$.
\end{lemma}

In this correspondence we always assume $d=\gcd(n,k)$. Then the
field $\bF_q$ is a vector space over $\bF_{q_0}$ with dimension $s$.
We fix a basis $v_1,\cdots,v_s$ of $\bF_q$ over $\bF_{q_0}$. Then
each $x\in \bF_q$ can be uniquely expressed as
\[x=x_1v_1+\cdots+x_sv_s\quad (x_i\in \bF_{q_0}).\]
Thus we have the following $\bF_{q_0}$-linear isomorphism:
\[\bF_q\xrightarrow{\sim}\bF_{q_0}^s,\quad x=x_1v_1+\cdots+x_sv_s\mapsto
X=(x_1,\cdots,x_s).\] With this isomorphism, a function $f:\bF_q\ra
\bF_{q_0}$ induces a function $F:\bF_{q_0}^s\ra \bF_{q_0}$ where for
$X=(x_1,\cdots,x_s)\in \bF_{q_0}^s, F(X)=f(x)$ with
$x=x_1v_1+\cdots+x_sv_s$. In this way, function
$f(x)=\Tra_{d}^n(\gamma x)$ for $\gamma\in \bF_q$ induces a linear
form \begin{equation} F(X)=\Tra_{d}^n(\gamma
x)=\sum\limits_{i=1}^{s}\Tra_{d}^n(\gamma v_i)x_i=A_{\ga}X^T
\end{equation}\label{def A_gamma}
 where $A_{\ga}=\left(\Tra_{d}^n(\gamma
v_1),\cdots,\Tra_{d}^n(\gamma v_s)\right),$
 and
$f_{\alpha,\beta}(x)=\Tra_d^n(\alpha x^{p^{3k}+1}+\beta x^{p^k+1})$
for $\al,\be\in \bF_q$ induces a quadratic form

\begin{eqnarray}\label{def H_al be}
F_{\alpha,\beta}(X)&=&\Tra_d^n(\alpha x^{p^{3k}+1}+\beta
x^{p^k+1})\nonumber\\[5mm]
&=&\Tra_{d}^n\left(\alpha\left(\sum\limits_{i=1}^s
x_iv_i^{p^{3k}}\right)\left(\sum\limits_{i=1}^s
x_iv_i\right)+\left(\beta\sum\limits_{i=1}^s
x_iv_i^{p^k}\right)\left(\sum\limits_{i=1}^s
x_iv_i\right)\right)\nonumber\\[2mm]
&=&\sum\limits_{i,j=1}^s\Tra_d^n\left(\al v_i^{p^{3k}}v_j+\be
v_i^{p^k}v_j\right)x_ix_j=XH_{\alpha,\beta}X^T
\end{eqnarray}

 where \[H_{\alpha,\beta}=(h_{ij})_{s\times s}\;\hbox{and}\;
h_{ij}=\frac{1}{2}\Tra_{d}^n\left(\alpha \left(v_i^{p^{3k}}v_j+
v_iv_j^{p^{3k}}\right)+\be \left(v_i^{p^k}v_j+
v_iv_j^{p^k}\right)\right)\;\hbox{for}\;1\leq i,j\leq s.\]

From Lemma \ref{qua},  in order to determine the values of
\[T(\alpha,\beta)=\sum\limits_{x\in \bF_q}\zeta_p^{\Tra_1^n(\alpha x^{p^{3k}+1}+\beta
x^{p^k+1})}=\sum\limits_{X\in
\bF_{q_0}^s}\zeta_p^{\Tra_1^{d}\left(XH_{\alpha,\beta}X^T\right)}\]
and
\[S(\alpha,\beta,\ga)=\sum\limits_{x\in \bF_q}\zeta_p^{\Tra_1^n(\alpha x^{p^{3k}+1}+\beta x^{p^k+1}+\ga
x)}=\sum\limits_{X\in
\bF_p^m}\zeta_p^{\Tra_1^{d}\left(XH_{\alpha,\beta}X^T+A_{\ga}X^T\right)}\quad
(\alpha,\beta,\ga\in \bF_q),\] we need to determine the rank of
$H_{\alpha,\beta}$ over $\bF_{q_0}$ and the solvability of
$\bF_{q_0}$-linear equation $2XH_{\al,\be}+A_{\ga}=0$.

Define $d'=\gcd(n,2k)$. Then an easy observation shows
\begin{equation}\label{rel d d'}
d'=\left\{
\begin{array}{ll}
2d, & \text{if}\; n/d\;\text{is even};\\[1mm]
d, &\text{otherwise.}
\end{array}
\right.
\end{equation}

Now we could determine the possible ranks of $H_{\al,\be}$.
\begin{lemma}\label{rank}
For $\alpha,\beta\in \bF_q$ and $(\al,\be)\neq \{(0,0)\}$, let
$r_{\alpha,\beta}$ be the rank of $H_{\alpha,\beta}$.  Then we have
\begin{itemize}
  \item[(i).] if $d'=d$, then the possible values of $r_{\al,\be}$
  are $s$, $s-1$, $s-2$.
  \item[(ii).] if $d'=2d$, then the possible values of $r_{\al,\be}$
  are $s$, $s-2$, $s-4$, $s-6$.
\end{itemize}
\end{lemma}
\begin{proof}
For (i), see \cite{Zen Hu Jia Yue Cao}. For (ii), see
\textbf{Appendix A.}
\end{proof}

 In order to determine the value distribution of
$T(\al,\be)$
 for $\al,\be\in \bF_q$, we need the
following result on moments of $T(\al,\be)$.
\begin{lemma}\label{moment}
For the exponential sum $T(\al,\be)$,
\[\begin{array}{ll}&(i). \;\;\sum\limits_{\al,\be\in
\bF_q}T(\al,\be)=p^{2n};\\[2mm]
                   &(ii).
                   \sum\limits_{\al,\be\in
\bF_q}T(\al,\be)^2=\left\{\begin{array}{ll}
p^{2n} &\text{if}\; d'=d\;\text{and}\;\; p^d\equiv 3\pmod 4,\\[1mm]
(2p^n-1)\cdot
p^{2n}&\text{if}\; d'=d\;\text{and}\; p^d\equiv 1\pmod 4,\\[1mm]
(p^{n+d}+p^n-p^d)\cdot p^{2n}&\text{if}\; d'=2d;
\end{array}\right.\\[2mm]
                   &(iii). \;\;\text{if}\;d'=2d,\;\text{then}\\[2mm]
                   &\quad\quad\quad \sum\limits_{\al,\be\in
\bF_q}T(\al,\be)^3=(p^{n+3d}+p^n-p^{3d})\cdot p^{2n}.\\[4mm]
\end{array}\]
\end{lemma}
\begin{proof}
see  \textbf{Appendix A.}
\end{proof}

In the case $d'=2d$, we could determine the explicit values of
$T(\al,\be)$. To this end we will study a class of Artin-Schreier
curves. A similar technique has been applied in Coulter \cite{Coul1},
Theorem 6.1.

\begin{lemma}\label{Artin}Suppose $(\al,\be)\in (\bF_{q}\times
\bF_q)\big{\backslash}\{0,0\}$ and $d'=2d$. Let $N$ be the number of
$\bF_q$-rational (affine) points on the curve
\begin{equation}\label{Artin Sch}
\al x^{p^{3k}+1}+\be x^{p^k+1}=y^{p^d}-y.
\end{equation}
 Then
\[N=q+(p^d-1)\cdot T(\al,\be).\]
\end{lemma}
\begin{proof}
see \textbf{Appendix A.}
\end{proof}

Now we give an explicit evaluation of $T(\al,\be)$ in the case
$d'=2d$.
\begin{lemma}\label{reduce num}
Assumptions as in Lemma \ref{Artin} and let $n=2m$, then
\[
T(\al,\be)=\left\{
\begin{array}{ll}
\mu p^{m}, &\text{if}\; r_{\al,\be}=s\\[2mm]
-\mu p^{m+d}, &\text{if}\; r_{\al,\be}=s-2\\[2mm]
\mu p^{m+2d}, &\text{if}\; r_{\al,\be}=s-4\\[2mm]
-\mu p^{m+3d}, &\text{if}\; r_{\al,\be}=s-6.
\end{array}
\right.
\]
where $\mu=(-1)^{m/d}$.
\end{lemma}
\begin{proof}
Consider the $\bF_q$-rational (affine) points on the Artin-Schreier
curve in Lemma \ref{Artin}. It is easy to verify that $(0,y)$ with
$y\in \bF_{p^d}$ are exactly the points on the curve with $x=0$. If
$(x,y)$ with $x\neq 0$ is a point on this curve, then so are $(t
x,t^{p^d+1}y)$ with $t^{p^{2d}-1}=1$ (note that $p^{3k}+1\equiv
p^k+1\equiv p^d+1\pmod {p^{2d}-1}$ since $3k/d$ and $k/d$ are both
odd by (\ref{rel d d'})). In total, we have
\[q+(p^d-1)T(\al,\be)=N\equiv p^d\pmod {p^{2d}-1}\]
which yields
\[T(\al,\be)\equiv 1\pmod{p^d+1}.\]

 We only consider the case $r_{\al,\be}=s$ and $m/d$ is odd. The other cases are similar. In this case $T(\al,\be)=\pm
 p^{m}$. Assume $T(\al,\be)=p^{m}$.
 Then ${p^d+1}\mid p^{m}-1$ which contradicts to $m/d$ is odd. Therefore $T(\al,\be)=-p^{m}$.
\end{proof}
\begin{remark}
(i). Our treatment improve the technique in \cite{Coul1}, \cite{Coul2}, in which
the case $(p,d)=(3,1)$ is dealt with in a different manner.

(ii). Applying Lemma \ref{reduce num} to Lemma \ref{Artin}, we could
determine the number of rational points on the curve (\ref{Artin
Sch}).
\end{remark}

\section{Exponential Sums $T(\al,\be)$ and Cyclic Code $\cC_1$}

\quad Define
\[N_{i}=\left\{(\al,\be)\in
\bF_{q}\times
\bF_q\backslash\{(0,0)\}\left|r_{\al,\be}=s-i\right.\right\}.
\] Then $n_i=\big{|}N_i\big{|}$ for $i=0,2,4,6$.

According to the possible values of $T(\al,\be)$ given by Lemma
\ref{qua} (setting $F(X)=XH_{\al,\be}X^T=\Tra_{d}^n(\al
x^{p^{3k}+1}+\be x^{p^k+1})$), we define that for $\varepsilon=\pm
1$,
\[N_{\varepsilon,i}=\left\{
\begin{array}{ll}
&\left\{(\al,\be)\in
\bF_q^2\backslash\{(0,0)\}\left|T(\al,\be)=\veps
p^{\frac{n+id}{2}} \right.\right\} \qquad\text{if}\; n+id\;\text{is even}\;,\\[3mm]
&\left\{(\al,\be)\in
\bF_q^2\backslash\{(0,0)\}\left|T(\al,\be)=\veps\sqrt{p^*}
p^{\frac{n+id-1}{2}} \right.\right\} \qquad\text{if}\;
n+id\;\text{is odd}
\end{array}
\right.\] where $p^*=(-1)^{\frac{p-1}{2}}p$  and
$n_{\veps,i}=|N_{\veps,i}|$.

Recall $q_0^*=(-1)^{\frac{q_0-1}{2}}q_0$.
 In this section we prove the following results.

\begin{theo}\label{value dis T}
The value distribution of the multi-set
$\left\{T(\al,\be)\left|\al,\be\in \bF_q\right.\right\}$ is shown as
following.

(i). For the case $d'=d$,
\begin{center}
\begin{tabular}{|c|c|}
\hline
values & multiplicity \\[2mm]
\hline $\sqrt{q_0^*}{q_0}^{\frac{s-1}{2}},
-\sqrt{q_0^*}{q_0}^{\frac{s-1}{2}}$&$\frac{1}{2}p^{2d}(p^n-p^{n-d}-p^{n-2d}+1)({p^n-1})/({p^{2d}-1})$
\\[2mm]
\hline
${p}^{\frac{n+d}{2}}$&$\frac{1}{2}p^{\frac{n-d}{2}}(p^{\frac{n-d}{2}}+1)(p^n-1)$
\\[2mm]
\hline

$-{p}^{\frac{n+d}{2}}$&$\frac{1}{2}{p}^{\frac{n-d}{2}}({p}^{\frac{n-d}{2}}-1)(p^n-1)$
\\[2mm]
\hline $\sqrt{{q_0}^*}{q_0}^{\frac{s+1}{2}},
-\sqrt{{q_0}^*}{q_0}^{\frac{s+1}{2}}$&$\frac{1}{2}(p^n-1)(p^{n-d}-1)/({p^{2d}-1})$
\\[2mm]

\hline $p^n$&$1$
\\[2mm]
\hline
\end{tabular}
\end{center}

(ii). For the case $d'=2d$,
\begin{center}
\begin{tabular}{|c|c|}
\hline
values & multiplicity \\[2mm]
\hline
$\mu p^m$&
$
\frac{(p^n-1)\left(p^{n+6d}-p^{n+4d}-p^{n+d}+\mu p^{m+5d}-\mu p^{m+4d}+p^{6d}\right)}{(p^d+1)(p^{2d}-1)(p^{3d}+1)}
$
\\[2mm]
\hline
$-\mu {p}^{m+d}$&
$
\frac{(p^n-1)(p^{n+3d}+p^{n+2d}-p^n-p^{n-d}-p^{n-2d}-\mu p^{m+3d}+\mu p^{m}+p^{3d})}{(p^d+1)^2(p^{2d}-1)}
$
\\[2mm]
\hline
 $\mu {p}^{m+2d}$&$
 \frac{
(p^{m-d}+\mu)(p^{m+d}+p^m-p^{m-2d}-\mu p^d)(p^{n}-1)}{(p^d+1)^3(p^{d}-1)}
$
\\[2mm]
\hline $-\mu {p}^{m+3d}$&$
\frac{(p^{m-2d}-\mu)(p^{m-d}+\mu)(p^{n}-1)}{(p^d+1)(p^{2d}-1)(p^{3d}+1)}
$
\\[2mm]

\hline $p^n$&$1$
\\[2mm]
\hline
\end{tabular}
\end{center}
where $\mu=(-1)^{m/d}$.
\end{theo}

\begin{proof}
see \cite{Zen Hu Jia Yue Cao} for (i) and \textbf{Appendix B} for
(ii).
\end{proof}
\begin{remark}

\begin{itemize}
\item[(i).] In the case $k\in\left\{\frac{n}{4},\frac{n}{2},\frac{3n}{4}\right\}$, the exponential sum $T(\al,\be)=\sum\limits_{x\in
\bF_q}\chi((\al^{p^k}+\be)x^{p^{k}+1})$ has been extensively studied, for example, see
\cite{Coul1}, \cite{Lid Nie}.
\item[(ii).] In the case $k\in\left\{\frac{n}{6},\frac{5n}{6}\right\}$, the exponential sum $T(\al,\be)=\sum\limits_{x\in \bF_q}\chi(\al x^{p^{m}+1}+\be x^{p^{k}+1})$
is a special case in \cite{Luo Tan Wan}. In the case $k\in \{\frac{n}{3},\frac{2n}{3}\}$, the exponential sum
 $T(\al,\be)=\sum\limits_{x\in \bF_q}\chi(\al x^{2}+\be x^{p^{k}+1})$
is a special case in \cite{Luo Fen}.
\end{itemize}
\end{remark}

Recall that $t$ is a divisor of $d$ and $\cC_1$ is the cyclic code
over $\bF_{p^t}$ with parity-check polynomial $h_2(x)h_3(x)$ where
$h_2(x)$ and $h_3(x)$ are the minimal polynomials of
$\pi^{-(p^{3k}+1)}$ and $\pi^{-(p^k+1)}$, respectively.
\begin{theo}\label{wei dis C1}
Suppose that $1\leq k\leq n-1$ and $k\notin\left\{\frac{n}{4},\frac{n}{2},\frac{3n}{4}\right\}$. Then the weight distribution
$\{A_0,A_1,\cdots,A_{q-1}\}$ of the cyclic code $\cC_1$ over
$\bF_{p^t}$ ($p\geq 3$) with length $q-1$ is shown as following.

(i). For the case $d=d'$ and $d/t$ are both odd, then
$\mathrm{dim}_{\bF_{p^t}}\cC_1=2n_0$ and $A_i=0$ except for
\begin{center}
\begin{tabular}{|c|c|}
\hline
$i$ & $A_i$ \\[2mm]
\hline
$(p^t-1)(p^{n-t}-p^{\frac{n+d}{2}-t})$&$\frac{1}{2}p^{\frac{n-d}{2}}(p^{\frac{n-d}{2}}+1)({p^n-1})$
\\[2mm]
\hline $(p^t-1)p^{n-t}$&$(p^n-1)(p^n-p^{n-d}+1)$
\\[2mm]
\hline

$(p^t-1)(p^{n-t}+p^{\frac{n+d}{2}-t})$&$\frac{1}{2}p^{\frac{n-d}{2}}(p^{\frac{n-d}{2}}-1)({p^n-1})$\\[2mm]
\hline $0$&$1$\\[1mm]
 \hline
\end{tabular}
\end{center}

(ii). For the case $d=d'$ and $d/t$ is even, then $\mathrm{dim}_{\bF_{p^t}}\cC_1=2n_0$ and $A_i=0$ except for
\begin{center}
\begin{tabular}{|c|c|}
\hline
$i$ & $A_i$ \\[2mm]
\hline
$(p^t-1)(p^{n-t}-p^{\frac{n}{2}+d-t})$&$\frac{1}{2}({p^n-1})(p^{n-d}-1)\big{\slash}(p^{2d}-1)$
\\[2mm]
\hline

$(p^t-1)(p^{n-t}-p^{\frac{n+d}{2}-t})$&$\frac{1}{2}p^{\frac{n-d}{2}}(p^{\frac{n-d}{2}}+1)({p^n-1})$\\[2mm]
\hline

$(p^t-1)(p^{n-t}-p^{\frac{n}{2}-t})$&$\frac{1}{2}p^{2d}(p^n-p^{n-d}-p^{n-2d}+1)({p^n-1})\big{\slash}(p^{2d}-1)$\\[2mm]
\hline

$(p^t-1)(p^{n-t}+p^{\frac{n}{2}-t})$&$\frac{1}{2}p^{2d}(p^n-p^{n-d}-p^{n-2d}+1)({p^n-1})\big{\slash}(p^{2d}-1)$\\[2mm]
\hline

$(p^t-1)(p^{n-t}+p^{\frac{n+d}{2}-t})$&$\frac{1}{2}p^{\frac{n-d}{2}}(p^{\frac{n-d}{2}}-1)({p^n-1})$\\[2mm]

\hline
$(p^t-1)(p^{n-t}+p^{\frac{n}{2}+d-t})$&$\frac{1}{2}({p^n-1})(p^{n-d}-1)\big{\slash}(p^{2d}-1)$
\\[2mm]
\hline $0$&$1$\\[1mm]

 \hline
\end{tabular}
\end{center}

(iii). For the case $d'=2d$ and $k\notin\left\{\frac{n}{6},\frac{5n}{6}\right\}$, then $\mathrm{dim}_{\bF_{p^t}}\cC_1=2n_0$ and $A_i=0$ except for
\begin{center}
\begin{tabular}{|c|c|}
\hline
$i$ & $A_i$ \\[2mm]
\hline
$(p^t-1)(p^{n-t}+\mu p^{m+3d-t})$&$
\frac{(p^{m-2d}-\mu)(p^{m-d}+\mu)(p^{n}-1)}{(p^d+1)(p^{2d}-1)(p^{3d}+1)}
$
\\[2mm]
\hline

$(p^t-1)(p^{n-t}-\mu {p}^{m+2d-t})$&$
 \frac{
(p^{m-d}+\mu)(p^{m+d}+p^m-p^{m-2d}-\mu p^d)(p^{n}-1)}{(p^d+1)^3(p^{d}-1)}
$
\\[2mm]
\hline

$(p^t-1)(p^{n-t}+\mu p^{m+d-t})$&$
\frac{
(p^n-1)(p^{n+3d}+p^{n+2d}-p^n-p^{n-d}-p^{n-2d}-\mu p^{m+3d}+\mu p^{m}+p^{3d})}{(p^d+1)^2(p^{2d}-1)}
$
\\[2mm]
\hline

$(p^t-1)(p^{n-t}-\mu p^{m-t})$&$
\frac{
(p^n-1)(p^{n+6d}-p^{n+4d}-p^{n+d}+\mu p^{m+5d}\mu p^{m+4d}+p^{6d})}{(p^d+1)(p^{2d}-1)(p^{3d}+1)}
$\\[2mm]
\hline
$0$&$1$\\[1mm]
 \hline
\end{tabular}
\end{center}

(iii). For the case $d'=2d$ and $k\in\left\{\frac{n}{6},\frac{5n}{6}\right\}$, then $\mathrm{dim}_{\bF_{p^t}}\cC_1=3n_0/2$ and $A_i=0$ except for
\begin{center}
\begin{tabular}{|c|c|}
\hline
$i$ & $A_i$ \\[2mm]
\hline
$(p^t-1)(p^{n-t}+{p}^{\frac{5n}{6}-t})$&$
(p^n-1)/(p^{\frac{n}{6}}+1)
$
\\[2mm]
\hline
$(p^t-1)(p^{n-t}-p^{\frac{2n}{3}-t})$&$p^{\frac{n}{6}}(p^{\frac{n}{3}}+1)(p^n-1)/(p^{\frac{n}{6}}+1)$
\\[2mm]
\hline

$(p^t-1)(p^{n-t}+p^{\frac{n}{2}-t})$&$
p^{\frac{n}{2}}(p^{\frac{n}{2}}-1)(p^{\frac{2n}{3}}-1)/(p^{\frac{n}{6}}+1)
$\\[2mm]
\hline
$0$&$1$\\[1mm]
 \hline
\end{tabular}
\end{center}
\end{theo}
\begin{proof}
see \textbf{Appendix B.}
\end{proof}
\begin{remark}
\begin{itemize}
    \item[(1).] In the case $d=d'$. Since $\gcd(p^n-1,
p^k+1)=2$, the first $l'=\frac{q-1}{2}$ coordinates of each codeword
of $\cC_1$ form a cyclic code $\cC_1'$ over $\bF_{p^t}$ with length
$l'$ and dimension $2n_0$. Let $(A_0',\cdots,A_{l'}')$ be the weight
distribution of $\cC_1'$, then $A_i'=A_{2i}$ $(0\leq i\leq l')$.
    \item[(2).] In the case $d'=2d$. Since $\gcd(p^{n}-1,
p^k+1)=p^d+1$, the first $l'=\frac{q-1}{p^d+1}$ coordinates of each
codeword of $\cC_1$ form a cyclic code $\cC_1'$ over $\bF_{p^t}$
with length $l'$ and dimension $2n_0$. Let $(A_0',\cdots,A_{l'}')$
be the weight distribution of $\cC_1'$, then $A_i'=A_{(p^d+1)i}$
$(0\leq i\leq l')$.
   \item[(3).] If $k\in\left\{\frac{n}{4},\frac{n}{2},\frac{3n}{4}\right\}$, then $\pi^{-(p^{3k}+1)p^k}=\pi^{-(p^k+1)}$ and the dual code of $\calC_1$ has
only one zero. This special case is trivial.
\end{itemize}
\end{remark}

\section{Results on Correlation Distribution of Sequences and Cyclic Code $\cC_2$}

\quad Recall $\phi_{\al,\be}(x)$ in the proof of Lemma \ref{rank}
and $N_{i,\veps}$ in the proof of Theorem \ref{value dis T}. Finally
we will determine the value distribution of $S(\al,\be,\ga)$, the
correlation distribution among sequences in $\calF$ defined in
(\ref{def F}) and the weight distribution of $\cC_2$ defined in
Section 1. The following result will play an important role.
\begin{lemma}\label{last lemma}
Assume $t$ be a divisor of $d$. For any $a\in \bF_{p^t}$ and any
$(\al,\be)\in N_{i,\veps}$ with $\veps=\pm 1$, then the number of
elements $\ga\in \bF_q$ satisfying
\begin{itemize}
    \item[(i).] $\phi_{\al,\be}(x)+\ga=0$ is solvable(choose one solution, say
    $x_0$),
    \item[(ii).] $\Tra_t^n(\alpha x_0^{p^{3k}+1}+\beta x_0^{p^k+1})=a$
\end{itemize}
is
\[
\left\{
    \begin{array}{ll}
    p^{n-id-t}&\text{if}\; s-i\; \text{and}\; d/t\;\text{are both
    odd},\;\text{and}\; a=0,\\[2mm]
    p^{n-id-t}+\veps \eta'(a)p^{\frac{n-id-t}{2}}&\text{if}\; s-i\; \text{and}\; d/t\;\text{are both
    odd},\;\text{and}\; a\neq 0,\\[2mm]
    p^{n-id-t}+
    \veps (p^t-1)p^{\frac{n-id}{2}-t}&\text{if}\; s-i\; \text{or}\; d/t\;\text{is
    even},\;\text{and}\; a=0,\\[2mm]
    p^{n-id-t}-\veps p^{\frac{n-id}{2}-t}&\text{if}\; s-i\; \text{or}\; d/t\;\text{is
    even},\;\text{and}\; a\neq 0.
    \end{array}
\right.
\]
where $\eta'$ is the quadratic (multiplicative) character on
$\bF_{p^t}$.
\end{lemma}
\begin{proof}
see \textbf{Appendix C.}
\end{proof}

Let $p^*=(-1)^{\frac{p-1}{2}}p$ and $\left(\frac{\cdot}{p}\right)$
be the Legendre symbol.
 We are now ready to give the value
distribution of $S(\al,\be,\ga)$.

\begin{theo}\label{value dis S}
The value distribution of the multi-set
$\left\{S(\al,\be,\ga)\left|\al\in \bF_{p^m},(\be,\ga)\in
\bF_q^2\right.\right\}$ is shown as following.

 (i). If $d'=d$ is odd (that is $n$ is odd), then
\begin{center}
\begin{tabular}{|c|c|}
\hline
values & multiplicity \\[2mm]
\hline $
\veps\sqrt{p^*}p^{\frac{n-1}{2}}$&$\frac{1}{2}p^{n+2d-1}\frac{(p^n-p^{n-d}-p^{n-2d}+1)({p^n-1})}{{p^{2d}-1}}$
\\[2mm]
\hline $ \veps \zeta_p^j
\sqrt{p^*}p^{\frac{n-1}{2}}$&$\frac{1}{2}p^{2d}(p^{n-1}+\veps
\left(\frac{-j}{p}\right)p^{\frac{n-1}{2}})\frac{(p^n-p^{n-d}-p^{n-2d}+1)({p^n-1})}{p^{2d}-1}$
\\[2mm]
\hline $\veps
p^{\frac{n+d}{2}}$&$\frac{1}{2}p^{n-d-1}(p^{\frac{n-d}{2}}+\veps
(p-1))(p^{\frac{n-d}{2}}+\veps)(p^n-1)$
\\[2mm]
\hline $\veps \zeta_p^j
p^{\frac{n+d}{2}}$&$\frac{1}{2}p^{n-d-1}(p^{\frac{n-d}{2}}-\veps)(p^{\frac{n-d}{2}}+\veps)(p^n-1)$
\\[2mm]

\hline
 $\veps\sqrt{{p}^*}{p}^{\frac{n+2d-1}{2}}$&$\frac{1}{2}p^{n-2d-1}\frac{(p^n-1)(p^{n-d}-1)}{{p^{2d}-1}}$
\\[2mm]
\hline
 $\veps\zeta_p^j\sqrt{{p}^*}{p}^{\frac{n+2d-1}{2}}$&$\frac{1}{2}(p^{n-2d-1}+\veps \left(\frac{-j}{p}\right)p^{\frac{n-2d-1}{2}})\frac{(p^n-1)(p^{n-d}-1)}{{p^{2d}-1}}$
\\[2mm]
\hline $0$&$(p^{n}-1)(p^{2n-d}-p^{2n-2d}+p^{2n-3d}-p^{n-2d}+1)$
\\[2mm]
\hline $p^n$&$1$
\\[2mm]
\hline
\end{tabular}
\end{center}
where $\veps =\pm 1, 1\leq j\leq p-1$.

(ii). If $d'=d$ is even, then

\begin{center}
\begin{tabular}{|c|c|}
\hline
values & multiplicity \\[2mm]
\hline $ \veps
p^{\frac{n}{2}}$&$\frac{1}{2}p^{2d}(p^{n-1}+\veps
(p-1)p^{\frac{n}{2}-1})\frac{(p^n-p^{n-d}-p^{n-2d}+1)(p^n-1)}{{p^{2d}-1}}$
\\[2mm]
\hline $ \veps \zeta_p^j
p^{\frac{n}{2}}$&$\frac{1}{2}p^{2d}(p^{n-1}-\veps
p^{\frac{n}{2}-1})\frac{(p^n-p^{n-d}-p^{n-2d}+1)(p^n-1)}{p^{2d}-1}$
\\[2mm]
\hline $\veps
p^{\frac{n+d}{2}}$&$\frac{1}{2}p^{n-d-1}(p^{\frac{n-d}{2}}+\veps
(p-1))(p^{\frac{n-d}{2}}+\veps)(p^n-1)$
\\[2mm]
\hline $\veps \zeta_p^j
p^{\frac{n+d}{2}}$&$\frac{1}{2}p^{n-d-1}(p^{\frac{n-d}{2}}-\veps)(p^{\frac{n-d}{2}}+\veps)(p^n-1)$
\\[2mm]

\hline
 $\veps{p}^{\frac{n+2d}{2}}$&$\frac{1}{2}(p^{n-2d-1}+\veps (p-1)p^{\frac{n-2d}{2}-1})\frac{(p^n-1)(p^{n-d}-1)}{p^{2d}-1}$
\\[2mm]
\hline
 $\veps\zeta_p^j{p}^{\frac{n+2d}{2}}$&$\frac{1}{2}(p^{n-2d-1}-\veps p^{\frac{n-2d}{2}-1})\frac{(p^n-1)(p^{n-d}-1)}{p^{2d}-1}$
\\[2mm]
\hline $0$&$(p^{n}-1)(p^{2n-d}-p^{2n-2d}+p^{2n-3d}-p^{n-2d}+1)$
\\[2mm]
\hline $p^n$&$1$
\\[2mm]
\hline
\end{tabular}
\end{center}
where $\veps =\pm 1, 1\leq j\leq p-1$.

(iii). For the case $d'=2d$,
\begin{center}
\begin{tabular}{|c|c|}
\hline
values & multiplicity \\[2mm]
\hline $\mu p^m$& $
\frac{(p^n-1)(p^{n-1}+\mu (p-1)p^{m-1})(p^{n+6d}-p^{n+4d}-p^{n+d}
+\mu p^{m+5d}-\mu p^{m+4d}+p^{6d})}{(p^d+1)(p^{2d}-1)(p^{3d}+1)}
$
\\[2mm]
\hline $\mu \zeta_p^j p^m$& $
\frac{(p^n-1)(p^{n-1}-\mu p^{m-1})(p^{n+6d}-p^{n+4d}-p^{n+d}+\mu p^{m+5d}-\mu p^{m+4d}+p^{6d})}{(p^d+1)(p^{2d}-1)(p^{3d}+1)}
$
\\[2mm]

\hline $-\mu {p}^{m+d}$& $
\frac{(p^{n-2d-1}-\mu(p-1)p^{m-d-1})(p^n-1)
(p^{n+3d}+p^{n+2d}-p^n-p^{n-d}-p^{n-2d}-\mu p^{m+3d}
+\mu p^{m}+p^{3d})}{(p^d+1)^2(p^{2d}-1)}
$
\\[2mm]

\hline $-\mu \zeta_p^j{p}^{m+d}$& $
\frac{(p^{n-2d-1}+\mu p^{m-d-1})(p^n-1)(p^{n+3d}+p^{n+2d}-p^n
-p^{n-d}-p^{n-2d}-\mu p^{m+3d}+\mu p^{m}+p^{3d})}{(p^d+1)^2(p^{2d}-1)}
$
\\[2mm]

\hline
 $\mu {p}^{m+2d}$&$
\frac{(p^{m-d}+\mu)(p^{m+d}+p^m-p^{m-2d}-\mu p^d)
(p^{n-4d-1}+\mu(p-1)p^{m-2d-1})(p^{n}-1)}{(p^d+1)^2(p^{2d}-1)}
$
\\[2mm]
\hline
 $\mu\zeta_p^j{p}^{m+2d}$&$
 \frac{(p^{m-d}+\mu)(p^{m+d}+p^m-p^{m-2d}-\mu p^d)
(p^{n-4d-1}-\mu p^{m-2d-1})(p^{n}-1)}{(p^d+1)^2(p^{2d}-1)}
$
\\[2mm]
\hline $-\mu {p}^{m+3d}$&$
\frac{(p^{m-2d}-\mu)(p^{m-d}+\mu)(p^{n-6d-1}
-\mu(p-1)p^{m-3d-1})(p^{n}-1)}{(p^d+1)(p^{2d}-1)(p^{3d}+1)}
$
\\[2mm]
\hline $-\mu \zeta_p^j {p}^{m+3d}$&$
\frac{(p^{m-2d}-\mu)(p^{m-d}+\mu)
(p^{n-6d-1}+\mu p^{m-3d-1})(p^{n}-1)}{(p^d+1)(p^{2d}-1)(p^{3d}+1)}$
\\[2mm]
\hline $0$&$
\begin{array}{ll}
(p^n-1)\left(1-\mu p^{3m-d}-\mu p^{3m-8d}+p^{n-d}+\right.&\\[2mm]
\qquad\left.\frac{ p^{2n}+p^{2n-9d}
+\mu p^{3m-3d}-\mu p^{3m-5d} -p^{n-4d}-p^{n-6d}}{p^d+1}\right)&
\end{array}
$\\[2mm]

\hline $p^n$&$1$
\\[2mm]
\hline
\end{tabular}
\end{center}
for $1\leq j\leq p-1$.
\end{theo}
\begin{proof}
For (i) and (ii), see \cite{Zen Hu Jia Yue Cao}. For (iii), see
\textbf{Appendix C.}
\end{proof}

In order to give the correlation distribution among the sequences in
$\calF$, we need an easy observation.

\begin{lemma}\label{q-2}
For any given $\ga\in \bF_q^*$, when $(\al,\be)$ runs through
$\bF_{q}\times \bF_q$, the distribution of $S(\al,\be,\ga)$ is the
same as $S(\al,\be,1)$.
\end{lemma}

As a consequence of Theorem \ref{value dis T}, Theorem \ref{value
dis S} and Lemma \ref{q-2}, we could give the correlation
distribution amidst the sequences in $\calF$.

\begin{theo}\label{cor dis}
Let $1\leq k\leq n-1$ and $k\notin\left\{\frac{n}{6},\frac{n}{4},\frac{n}{2},\frac{3n}{4},\frac{5n}{6}\right\}$. The collection $\calF$ defined in (\ref{def F}) is a family of
$p^{2n}$ $p$-ary sequences with period $q-1$. Its correlation
distribution is given as follows.

 (i). For the case $d'=d$ is odd (that is $n$ is odd), then
\begin{center}
\begin{tabular}{|c|c|}
\hline
values & multiplicity \\[2mm]
\hline $
\veps\sqrt{p^*}p^{\frac{n-1}{2}}-1$&$\frac{1}{2}p^{2n+2d}({p^{2n-1}-2p^{n-1}+1})\frac{p^n-p^{n-d}-p^{n-2d}+1}{p^{2d}-1}$
\\[2mm]
\hline $ \veps \zeta_p^j
\sqrt{p^*}p^{\frac{n-1}{2}}-1$&$\frac{1}{2}p^{2n+2d}(p^{n-1}+\veps
\left(\frac{-j}{p}\right)p^{\frac{n-1}{2}})({p^n-2})\frac{p^n-p^{n-d}-p^{n-2d}+1}{p^{2d}-1}$
\\[2mm]
\hline $\veps
p^{\frac{n+d}{2}}-1$&$\frac{1}{2}p^{\frac{5n-d}{2}}(p^{\frac{n-d}{2}}+\veps)\left(p^{\frac{n-d}{2}-1}(p^{\frac{n-d}{2}}+\veps(p-1))(p^n-2)+1\right)$
\\[2mm]
\hline $\veps \zeta_p^j
p^{\frac{n+d}{2}}-1$&$\frac{1}{2}p^{3n-d-1}(p^{\frac{n-d}{2}}-\veps)(p^{\frac{n-d}{2}}+\veps)(p^n-2)$
\\[2mm]

\hline
 $\veps\sqrt{{p}^*}{p}^{\frac{n+2d-1}{2}}-1$&$\frac{1}{2}p^{2n}(p^{2n-2d-1}-2p^{n-2d-1}+1)\frac{p^{n-d}-1}{p^{2d}-1}$
\\[2mm]
\hline
 $\veps\zeta_p^j\sqrt{{p}^*}{p}^{\frac{n+2d-1}{2}}-1$&$\frac{1}{2}p^{2n}(p^{n-2d-1}+\veps \left(\frac{-j}{p}\right)p^{\frac{n-2d-1}{2}})(p^n-2)\frac{p^{n-d}-1}{p^{2d}-1}$
\\[2mm]
\hline $-1$&$p^{2n}(p^{n}-2)(p^{2n-d}-p^{2n-2d}+p^{2n-3d}-p^{n-2d}+1)$
\\[2mm]
\hline $p^n-1$&$p^{2n}$
\\[2mm]
\hline
\end{tabular}
\end{center}
where $\veps =\pm 1, 1\leq j\leq p-1$.

(ii). For the case $d'=d$ is even, then

\begin{center}
\begin{tabular}{|c|c|}
\hline
values & multiplicity \\[2mm]
\hline $ \veps
p^{\frac{n}{2}}-1$&$\frac{1}{2}p^{2n+2d}\left((p^{n-1}+\veps
(p-1)p^{\frac{n}{2}-1})({p^n-2})+1\right)\frac{p^n-p^{n-d}-p^{n-2d}+1}{p^{2d}-1}$
\\[2mm]
\hline $ \veps \zeta_p^j
p^{\frac{n}{2}}-1$&$\frac{1}{2}p^{2n+2d}(p^{n-1}-\veps
p^{\frac{n}{2}-1})({p^n-2})\frac{p^n-p^{n-d}-p^{n-2d}+1}{p^{2d}-1}$
\\[2mm]
\hline $\veps
p^{\frac{n+d}{2}}-1$&$\frac{1}{2}p^{\frac{5n-d}{2}}(p^{\frac{n-d}{2}}+\veps)\left(p^{\frac{n-d}{2}}(p^{\frac{n-d}{2}}+\veps
(p-1))(p^n-2)+1\right)$
\\[2mm]
\hline $\veps \zeta_p^j
p^{\frac{n+d}{2}}-1$&$\frac{1}{2}p^{5n-d-1}(p^{\frac{n-d}{2}}-\veps)(p^{\frac{n-d}{2}}+\veps)(p^n-2)$
\\[2mm]

\hline
 $\veps{p}^{\frac{n+2d}{2}}-1$&$\frac{1}{2}p^{2n}\left(\left(p^{n-2d-1}+\veps (p-1)p^{\frac{n-2d}{2}-1}\right)(p^n-2)+1\right)\frac{p^{n-d}-1}{p^{2d}-1}$
\\[2mm]
\hline
 $\veps\zeta_p^j{p}^{\frac{n+2d}{2}}-1$&$\frac{1}{2}p^{2n}(p^{n-2d-1}-\veps p^{\frac{n-2d}{2}-1})(p^n-2)(p^{n-d}-1)/({p^{2d}-1})$
\\[2mm]
\hline $-1$&$p^{2n}(p^{n}-2)(p^{2n-d}-p^{2n-2d}+p^{2n-3d}-p^{n-2d}+1)$
\\[2mm]
\hline $p^n-1$&$p^{2n}$
\\[2mm]
\hline
\end{tabular}
\end{center}
where $\veps =\pm 1, 1\leq j\leq p-1$.

(iii). For the case $d'=2d$,
\begin{center}
\begin{tabular}{|c|c|}
\hline
values & multiplicity \\[2mm]
\hline $\mu p^m-1$& $
\frac{p^{2n}\left((p^n-2)(p^{n-1}+\mu(p-1)p^{m-1})+1\right)\left(p^{n+6d}-p^{n+4d}-p^{n+d}
+\mu p^{m+5d}-\mu p^{m+4d}+p^{6d}\right)}{(p^d+1)(p^{2d}-1)(p^{3d}+1)}
$
\\[2mm]
\hline $\mu\zeta_p^j p^m-1$& $
\frac{p^{2n}(p^n-2)(p^{n-1}-\mu p^{m-1})(p^{n+6d}-p^{n+4d}-p^{n+d}+\mu p^{m+5d}-\mu p^{m+4d}+p^{6d})}{(p^d+1)(p^{2d}-1)(p^{3d}+1)}
$
\\[2mm]

\hline $-\mu {p}^{m+d}-1$& $
\frac{p^{2n}\left((p^{n-2d-1}-\mu(p-1)p^{m-d-1})(p^n-2)+1\right)(p^{n+3d}+p^{n+2d}-p^n-p^{n-d}-p^{n-2d}-\mu p^{m+3d}
+\mu p^{m}+p^{3d})}{(p^d+1)^2(p^{2d}-1)}
$
\\[2mm]

\hline $-\mu \zeta_p^j{p}^{m+d}-1$& $
\frac{
p^{2n}(p^{n-2d-1}+\mu p^{m-d-1})(p^n-2)(p^{n+3d}+p^{n+2d}-p^n-p^{n-d}
-p^{n-2d}-\mu p^{m+3d}+\mu p^{m}+p^{3d})}{(p^d+1)^2(p^{2d}-1)}
$
\\[2mm]

\hline
 $\mu {p}^{m+2d}-1$&$
\frac{
p^{2n}\left((p^{m-d}+\mu)(p^{n}-2)+1\right)(p^{m+d}+p^m-p^{m-2d}-\mu p^d)
(p^{n-4d-1}+\mu (p-1)p^{m-2d-1})}{(p^d+1)^2(p^{2d}-1)}
$
\\[2mm]
\hline
 $\mu \zeta_p^j{p}^{m+2d}-1$&$
\frac{
p^{2n}(p^{m-d}+\mu)(p^{m+d}+p^m-p^{m-2d}-\mu p^d)
(p^{n-4d-1}-\mu p^{m-2d-1})(p^{n}-2)}{(p^d+1)^2(p^{2d}-1)}
$
\\[2mm]
\hline $-\mu {p}^{m+3d}-1$&$
\frac{
p^{2n}\left((p^{m-2d}-\mu)(p^{n}-2)+1\right)(p^{m-d}+\mu)(p^{n-6d-1}
-\mu(p-1)p^{m-3d-1})}{(p^d+1)(p^{2d}-1)(p^{3d}+1)}
$
\\[2mm]
\hline $-\mu \zeta_p^j {p}^{m+3d}-1$&$
\frac{
p^{2n}(p^{m-2d}-\mu)(p^{m-d}+\mu)
(p^{n-6d-1}+\mu p^{m-3d-1})(p^{n}-2)}{(p^d+1)(p^{2d}-1)(p^{3d}+1)}$
\\[2mm]
\hline $-1$&$
\begin{array}{ll}
p^{2n}(p^n-2)\left(1-\mu p^{3m-d}-\mu p^{3m-8d}+p^{n-d}\right.&\\[2mm]
\qquad\quad \left.+\frac{p^{2n}+p^{2n-9d}
+\mu p^{3m-3d}-\mu p^{3m-5d}-p^{n-4d}-p^{n-6d}}{p^d+1}\right)&
\end{array}
$\\[2mm]

\hline $p^n-1$&$p^{2n}$
\\[2mm]
\hline
\end{tabular}
\end{center}
for $1\leq j\leq p-1$ and $\mu=(-1)^{m/d}$.
\end{theo}

Recall that $\cC_2$ is the cyclic code over $\bF_{p^t}$ with
parity-check polynomial $h_1(x)h_2(x)h_3(x)$ where $h_1(x)$,
$h_2(x)$ and $h_3(x)$ are the minimal polynomials of $\pi^{-1}$,
$\pi^{-(p^k+1)}$ and $\pi^{-(p^{3k}+1)}$ respectively. Here we are
ready to determine the weight distribution of $\cC_2$.

\begin{theo}\label{wei dis C2}
For $n\geq 3$,$k\notin\left\{\frac{n}{6},\frac{n}{4},\frac{n}{2},\frac{3n}{4},\frac{5n}{6}\right\}$, the weight distribution
$\{A_0,A_1,\cdots,A_{q-1}\}$ of the cyclic code $\cC_2$ over $\bF_{p^t}$
($p\geq 3$) with length $q-1$ and
$\mathrm{dim}_{\bF_{p^t}}\cC_1=3n_0$ is shown as follows:

(i). the case $d'=d$ and $d/t$ is odd,
\begin{center}
\begin{tabular}{|c|c|}
\hline
$i$ & $A_i$ \\[2mm]
\hline
$(p^t-1)p^{n-t}-p^{\frac{n+2d-t}{2}}$&$\frac{1}{2}p^{\frac{n-2d-t}{2}}(p^t-1)(p^{\frac{n-2d-t}{2}}+1)\frac{(p^{n-d}-1)(p^n-1)}{p^{2d}-1}$
\\[2mm]
\hline

$(p^t-1)(p^{n-d}-p^{\frac{n+d-2t}{2}})$&$\frac{1}{2}p^{n-d-t}(p^{\frac{n-d}{2}}+1)(p^{\frac{n-d}{2}}+p^t-1)(p^n-1)$
\\[2mm]
\hline

$(p^t-1)p^{n-t}-p^{\frac{n+d-2t}{2}}$&$\frac{1}{2}p^{n-d-t}(p^t-1)(p^{n-d}-1)(p^n-1) $\\[2mm]
\hline

$(p^t-1)p^{n-t}-p^{\frac{n-t}{2}}$&$\frac{1}{2}p^{\frac{n-t}{2}+2d}(p^t-1)(p^{\frac{n-t}{2}}+1)(p^{n}-p^{n-d}-p^{n-2d}+1)\frac{p^n-1}{p^{2d}-1}$\\[2mm]
\hline

$(p^t-1)p^{n-t}$&$
\begin{array}{ll}
(p^n-1)(p^{2n-t}+p^{2n-d}-p^{2n-d-t}-p^{2n-2d}+p^{2n-3d}&\\[2mm]
\qquad-p^{2n-3d-t}+p^{n-t}-p^{n-2d}+p^{n-2d-t}+1)&
\end{array}
$\\[2mm]
\hline

$(p^t-1)p^{n-t}+p^{\frac{n-t}{2}}$&$\frac{1}{2}p^{\frac{n-t}{2}+2d}(p^t-1)(p^{\frac{n-t}{2}}-1)(p^{n}-p^{n-d}-p^{n-2d}+1)\frac{p^n-1}{p^{2d}-1}$\\[2mm]
\hline

$(p^t-1)p^{n-t}+p^{\frac{n+d-2t}{2}}$&$\frac{1}{2}p^{n-d-t}(p^t-1)(p^{n-d}-1)(p^n-1) $\\[2mm]
\hline

$(p^t-1)(p^{n-d}+p^{\frac{n+d-2t}{2}})$&$\frac{1}{2}p^{n-d-t}(p^{\frac{n-d}{2}}-1)(p^{\frac{n-d}{2}}-p^t+1)(p^n-1)$
\\[2mm]
\hline
$(p^t-1)p^{n-t}+p^{\frac{n+2d-t}{2}}$&$\frac{1}{2}p^{\frac{n-2d-t}{2}}(p^t-1)(p^{\frac{n-2d-t}{2}}-1)\frac{(p^{n-d}-1)(p^n-1)}{p^{2d}-1}$
\\[2mm]
\hline
$0$&$1$\\[2mm]
\hline
\end{tabular}
\end{center}

 (ii). the case $d'=d$  and $d/t$ is even,
\begin{center}
\begin{tabular}{|c|c|}
\hline
$i$ & $A_i$ \\[2mm]
\hline
$(p^t-1)(p^{n-t}-p^{\frac{n}{2}+d-t})$&$\frac{1}{2}p^{\frac{n}{2}-d-t}(p^{\frac{n}{2}-d}+p^t-1)(p^{n-d}-1)\frac{p^n-1}{p^{2d}-1}$
\\[2mm]
\hline
$(p^t-1)p^{n-t}-p^{\frac{n}{2}+d-t}$&$\frac{1}{2}p^{\frac{n}{2}-d-t}(p^t-1)(p^{\frac{n}{2}-d}+1)(p^{n-d}-1)\frac{p^n-1}{p^{2d}-1}$\\[2mm]
\hline

$(p^t-1)(p^{n-t}-p^{\frac{n+d}{2}-t})$&$\frac{1}{2}p^{n-d-t}(p^{\frac{n-d}{2}}+1)(p^{\frac{n-d}{2}}+p^t-1)(p^n-1)$\\[2mm]
\hline

$(p^t-1)p^{n-t}-p^{\frac{n+d}{2}-t}$&$\frac{1}{2}p^{n-d-t}(p^t-1)(p^{n-d}-1)(p^n-1)$\\[2mm]
\hline

\hline
$(p^t-1)(p^{n-t}-p^{\frac{n}{2}-t})$&$\frac{1}{2}p^{\frac{n}{2}+2d-t}(p^{\frac{n}{2}}+p^t-1)(p^n-p^{n-d}-p^{n-2d}+1)\frac{p^n-1}{p^{2d}-1}$
\\[2mm]
\hline

$(p^t-1)p^{n-t}-p^{\frac{n}{2}-t}$&$\frac{1}{2}p^{\frac{n}{2}+2d-t}(p^t-1)(p^{\frac{n}{2}}+1)(p^n-p^{n-d}-p^{n-2d}+1)\frac{p^n-1}{p^{2d}-1}$
\\[2mm]
\hline

$(p^t-1)p^{n-t}$&$(p^n-1)(p^{2n-d}-p^{2n-2d}+p^{2n-3d}-p^{n-2d}+1)$\\[2mm]
\hline

$(p^t-1)p^{n-t}+p^{\frac{n}{2}-t}$&$\frac{1}{2}p^{\frac{n}{2}+2d-t}(p^t-1)(p^{\frac{n}{2}}-1)(p^n-p^{n-d}-p^{n-2d}+1)\frac{p^n-1}{p^{2d}-1}$
\\[2mm]
\hline

$(p^t-1)(p^{n-t}+p^{\frac{n}{2}-t})$&$\frac{1}{2}p^{\frac{n}{2}+2d-t}(p^{\frac{n}{2}}-p^t+1)(p^n-p^{n-d}-p^{n-2d}+1)\frac{p^m-1}{p^{2d}-1}$
\\[2mm]
\hline

$(p^t-1)p^{n-t}+p^{\frac{n+d}{2}-t}$&$\frac{1}{2}p^{n-d-t}(p^t-1)(p^{n-d}-1)(p^n-1)$\\[2mm]
\hline

$(p^t-1)(p^{n-t}+p^{\frac{n+d}{2}-t})$&$\frac{1}{2}p^{n-d-t}(p^{\frac{n-d}{2}}-1)(p^{\frac{n-d}{2}}-p^t+1)(p^n-1)$\\[2mm]
\hline

$(p^t-1)p^{n-t}+p^{\frac{n}{2}+d-t}$&$\frac{1}{2}p^{\frac{n}{2}-d-t}(p^t-1)(p^{\frac{n}{2}-d}-1)(p^{n-d}-1)\frac{p^n-1}{p^{2d}-1}$\\[2mm]
\hline
$(p^t-1)(p^{n-t}+p^{\frac{n}{2}-d-t})$&$\frac{1}{2}p^{\frac{n}{2}+d-t}(p^{\frac{n}{2}-d}-p^t+1)(p^{n-d}-1)\frac{p^n-1}{p^{2d}-1}$
\\[2mm]
\hline

$0$&$1$\\[2mm]
\hline
\end{tabular}
\end{center}

(iii). the case $d'=2d$,
\begin{center}
\begin{tabular}{|c|c|}
\hline
$i$ & $A_i$ \\[2mm]

\hline
$(p^t-1)\left(p^{n-t}-\mu p^{m-t}\right)$&
$
\frac{\left(p^{n-t}+\mu(p^t-1) p^{m-t}\right)(p^n-1)\left(p^{n+6d}-p^{n+4d}-p^{n+d}+\mu p^{m+5d}-\mu p^{m+4d}+p^{6d}\right)}{(p^d+1)(p^{2d}-1)(p^{3d}+1)}$
\\[2mm]

\hline
$(p^t-1)p^{n-t}+\mu p^{m-t}$&
$
\frac{(p^t-1)(p^{n-t}-\mu p^{m-t})(p^n-1)\left(p^{n+6d}-p^{n+4d}-p^{n+d}+\mu p^{m+5d}-\mu p^{m+4d}+p^{6d}\right)}{(p^d+1)(p^{2d}-1)(p^{3d}+1)}
$
\\[2mm]

\hline
$(p^t-1)\left(p^{n-t}+\mu p^{m+d-t}\right)$&
$
\frac{\left(p^{n-2d-t}-\mu (p^t-1)p^{m-d-t}\right)(p^n-1)(p^{n+3d}+p^{n+2d}-p^n-p^{n-d}-p^{n-2d}-\mu p^{m+3d}+\mu p^{m}+p^{3d})}{(p^d+1)^2(p^{2d}-1)}
$
\\[2mm]
\hline
$(p^t-1)p^{n-t}-\mu p^{m+d-t}$&
$
\frac{(p^t-1)\left(p^{n-2d-t}+\mu p^{m-d-t}\right)(p^n-1)(p^{n+3d}+p^{n+2d}-p^n-p^{n-d}-p^{n-2d}-\mu p^{m+3d}+\mu p^{m}+p^{3d})}{(p^d+1)^2(p^{2d}-1)}
$
\\[2mm]
\hline

$(p^t-1)\left(p^{n-t}-\mu p^{m+2d-t}\right)$&$
\frac{(p^{n-4d-t}+\mu (p^t-1)p^{m-2d-t})(p^{m-d}+\mu)(p^{m+d}+p^m-p^{m-2d}-\mu p^d)
(p^{n}-1)}{(p^d+1)^3(p^{d}-1)}
$
\\[2mm]
\hline

$(p^t-1)p^{n-t}+\mu p^{m+2d-t}$&$
\frac{(p^t-1)(p^{n-4d-t}-\mu p^{m-2d-t})(p^{m-d}+\mu)(p^{m+d}+p^m-p^{m-2d}-\mu p^d)
(p^{n}-1)}{(p^d+1)^3(p^{d}-1)}
$
\\[2mm]
\hline

$(p^t-1)\left(p^{n-t}+\mu p^{m+3d-t}\right)$&$
\frac{\left(p^{n-6d-t}-\mu (p^t-1)p^{m-3d-t}\right)(p^{m-2d}-\mu)(p^{m-d}+\mu)(p^{n}-1)}{(p^d+1)(p^{2d}-1)(p^{3d}+1)}
$
\\[2mm]

\hline
$(p^t-1)p^{n-t}-\mu p^{m+3d-t}$&$
\frac{(p^t-1)\left(p^{n-6d-t}+\mu p^{m-3d-t}\right)(p^{m-2d}-\mu)(p^{m-d}+\mu)(p^{n}-1)}{(p^d+1)(p^{2d}-1)(p^{3d}+1)}
$
\\[2mm]

\hline

$(p^t-1)p^{n-t}$&$
\begin{array}{ll}
(p^{2n}+p^{2n-9d}-\mu p^{3m}+\mu
p^{3m-d}
+\mu  p^{3m-3d}-\mu p^{3m-5d} - \mu
p^{3m-7d}&\\[1mm]
+\mu
p^{3m-8d}+p^n-p^{n-d}-p^{n-4d}-p^{n-6d}+p^d+1)\frac{p^n-1}{p^d+1}&
\end{array}
$\\[2mm]

\hline $p^n$&$1$
\\[2mm]
\hline
\end{tabular}
\end{center}
where $\mu=(-1)^{\frac{m}{d}}$.

\end{theo}
\begin{proof}
see {\it\textbf{Appendix C.}}
\end{proof}
\begin{remark}
(i). The case (i) and (ii) with $t=1$ has been shown in \cite{Zen Hu Jia Yue Cao}, Theorem 2.

(ii). If $k=n/6$ or $5n/6$, then $\cC_2$ has dimension $5n_0/2$. Its weight distribution has been determined in
\cite{Luo Tan Wan}.
\end{remark}

\section{Appendix A}
{\it \textbf{Proof of Lemma \ref{rank}}(ii)}: For
$Y=(y_1,\cdots,y_s)\in \bF_{q_0}^s$, $y=y_1v_1+\cdots+y_sv_s\in
\bF_q$, we know that
\begin{equation}\label{bil form1}
F_{\alpha,\beta}(X+Y)-F_{\alpha,\beta}(X)-F_{\alpha,\beta}(Y)=2XH_{\alpha,\beta}Y^T
\end{equation}
is equal to
\begin{equation}\label{bil form2}
f_{\alpha,\beta}(x+y)-f_{\alpha,\beta}(x)-f_{\alpha,\beta}(y)=\Tra_{d}^n\left(y^{p^{3k}}(\alpha^{p^{3k}}
x^{p^{6k}}+\be^{p^{3k}} x^{p^{4k}}+\beta^{p^{2k}} x^{p^{2k}}+\al
x)\right)
\end{equation}

 Let
\begin{equation}\label{def phi}
\phi_{\al,\be}(x)=\alpha^{p^{3k}} x^{p^{6k}}+\be^{p^{3k}}
x^{p^{4k}}+\beta^{p^{2k}} x^{p^{2k}}+\al x.
\end{equation}
 Therefore,
\[{\setlength\arraycolsep{2pt}
\begin{array}{lcl}
r_{\al,\be}=r& \Leftrightarrow&\text{the number of common solutions of}\;XH_{\alpha,\beta}Y^T=0\;\text{for all}\;Y\in \bF_{q_0}^s\;\text{is}\; q_0^{s-r}, \\[2mm]
& \Leftrightarrow&\text{the number of common solutions of}\;\Tra_{d}^n\left(y^{3k}\cdot\phi_{\al,\be}(x)\right)=0\;\text{for all}\;y\in \bF_q\;\text{is}\; q_0^{s-r}, \\[2mm]
&\Leftrightarrow&\phi_{\al,\be}(x)=0\;\text{has}\; q_0^{s-r}\;
\text{solutions in}\; \bF_q.
\end{array}
}
\]

Fix an algebraic closure $\bF_{p^\infty}$ of $\bF_p$, since the
degree of $p^{2k}$-linearized polynomial $\phi_{\al,\be}(x)$ is
$p^{6k}$ and $\phi_{\al,\be}(x)=0$ has no multiple roots in
$\bF_{p^{\infty}}$ (this fact follows from
$\phi'_{\al,\be}(x)=\al\in \bF_q^*$),  then the zeroes of
$\phi_{\al,\be}(x)$ in $\bF_{p^\infty}$, say $V$, form an
$\bF_{p^{2k}}$-vector space of dimension 3. Note that
$\gcd(n,2k)=2d$. Then $V\cap \bF_{p^n}$ is a vector space on
$\bF_{p^{\gcd(n,2k)}}=\bF_{p^{2d}}$ with dimension at most 3 since
any elements in $\bF_q$ which are linear independent over
$\bF_{p^{2d}}$ are also linear independent over $\bF_{p^{2k}}$(see
\cite{Trac}, Lemma 4). Since $\bF_{p^{2d}}$ could be regarded as a
2-dimensional vector space over $\bF_{p^d}$, then the possible
values of $r_{\al,\be}$ is $s$, $s-2$, $s-4$ and $s-6$ for
$(\al,\be)\in \bF_q^2\backslash\{(0,0)\}$. $\square$

 {\it\textbf{Proof of Lemma \ref{moment}}}: (i). We observe that
\[ {\setlength\arraycolsep{2pt}
\begin{array}{ll}
&\sum\limits_{\al,\be\in
\bF_q}T(\al,\be)=\sum\limits_{\al,\be\in\bF_q}\sum\limits_{x\in
\bF_q}\zeta_p^{\Tra_1^n(\al x^{p^{3k}+1}+\be x^{p^k+1})}\\[3mm]
&\quad\quad=\sum\limits_{x\in\bF_q}\sum\limits_{\al\in
\bF_{q}}\zeta_p^{\Tra_1^n(\al x^{p^{3k}+1})}\sum\limits_{\be\in
\bF_q}\zeta_p^{\Tra_1^n(\be
x^{p^k+1})}=q\cdot\sum\limits_{\stackrel{\al\in
\bF_{q}}{x=0}}\zeta_p^{\Tra_1^n(\al x^{p^{3k}+1})}=p^{2n}.
\end{array}
}
\]
(ii). We can calculate
\[
{ \setlength\arraycolsep{2pt}
\begin{array}{lll}
\sum\limits_{\al,\be\in\bF_q}T(\al,\be)^2&=&\sum\limits_{x,y\in
\bF_q}\sum\limits_{\al\in
\bF_{q}}\zeta_p^{\Tra_1^n\left(\al\left(x^{p^{3k}+1}+y^{p^{3k}+1}\right)\right)}\sum\limits_{\be\in
\bF_q}\zeta_p^{\Tra_1^n\left(\be\left(x^{p^k+1}+y^{p^k+1}\right)\right)}\\[2mm]
&=&M_2\cdot p^{2n}\end{array}}
\] where
$M_2$ is the number of solutions to the equation
\begin{eqnarray}\label{def 2nd}
\left\{
\begin{array}{ll}
 x^{p^{3k}+1}+y^{p^{3k}+1}=0&\\[2mm]
 x^{p^k+1}+y^{p^k+1}=0&
 \end{array}
 \right.
\end{eqnarray}

If $xy=0$ satisfying (\ref{def 2nd}), then $x=y=0$. Otherwise
$(x/y)^{p^{3k}+1}=(x/y)^{p^k+1}=-1$ which yields that
$(x/y)^{p^{2k}-1}=1$. Denote by $x=ty$.  Since $\gcd(2k,n)=d'$, then
$t\in \bF_{p^{d'}}^*$.
\begin{itemize}
  \item If $d'=d$, then $t\in \bF_{p^d}^*$ and (\ref{def 2nd}) is
  equivalent to $x^2+y^2=0$. Hence $t^2=-1$. There are two or none
  of $t\in \bF_{p^d}^*$ satisfying $t^2=-1$ depending on $p^d\equiv
  1\pmod 4$ or $p^d\equiv
  3\pmod 4$. Therefore
  \[
  M_2=\left\{
  \begin{array}{ll}
  1+2(q-1)=2q-1, &\text{if}\; p^d\equiv
  1\pmod 4\\[2mm]
  1, &\text{if}\; p^d\equiv
  3\pmod 4.
  \end{array}
  \right.
  \]
  \item If $d'=2d$, then by (\ref{rel d d'}) we get (\ref{def 2nd}) is equivalent to
  $x^{p^d+1}+y^{p^d+1}=0$. Then we have $t^{p^d+1}=-1$ which has
  $p^d+1$ solutions in $\bF_{p^{d'}}^*$. Therefore
  \[M_2=(p^d+1)(p^n-1)+1=p^{n+d}+p^n-p^d.\]
\end{itemize}

(iii). We have
\[\sum\limits_{\al,\be\in\bF_q}T(\al,\be)^3=M_3\cdot q^2\quad \text{where}\]
\begin{eqnarray}\label{num M}
&&M_3=\#\left\{(x,y,z)\in
\bF_q^3\left|x^{p^{3k}+1}+y^{p^{3k}+1}+z^{p^{3k}+1}=0,x^{p^k+1}+y^{p^k+1}+z^{p^k+1}=0\right.\right\}\nonumber\\[2mm]
&&\quad\;=M_2+T'\cdot(q-1)
\end{eqnarray}
and $T'$ is the number of $\bF_q$-solutions of
\begin{equation}\label{def T'}
\left\{
\begin{array}{ll}
x^{p^{3k}+1}+y^{p^{3k}+1}+1=0&\\[2mm]
x^{p^k+1}+y^{p^k+1}+1=0.&
\end{array}
\right.
\end{equation}
Canceling $y$ we have
$\left(x^{p^{3k}+1}+1\right)^{p^k+1}=\left(x^{p^{k}+1}+1\right)^{p^{3k}+1}$
which is equivalent to
\[(x^{p^{4k}}-x)(x^{p^k}-x^{p^{3k}})=0.\]
Therefore $x^{p^{4k}}=x$ or $x^{p^k}=x^{p^{3k}}$.
\begin{itemize}
  \item If $n/d\equiv 2\pmod 4$, then $x\in \bF_{p^{2d}}$ and symmetrically $y\in \bF_{p^{2d}}$. Hence
   (\ref{def T'}) is equivalent
  to $x^{p^d+1}+y^{p^d+1}+1=0$ which is the well-known Hermitian curve on $\bF_{p^{2d}}$.
  If follows that $T'=p^{3d}-p^d$.
  \item If $n/d\equiv 0\pmod 4$, then $x\in \bF_{p^{4d}}$ and hence $y\in \bF_{p^{4d}}$. In this case
  $\left(x^{p^k+1}+y^{p^k+1}+1\right)^{p^{3k}}=x^{p^{3k}+1}+y^{p^{3k}+1}+1$ and then
   (\ref{def T'}) is
  equivalent to $x^{p^d+1}+y^{p^d+1}+1=0$ which is a minimal curve on $\bF_{p^{4d}}$ with genus $\frac{1}{2}p^d(p^d-1)$.
  Hence
  $T'=p^{4d}+1-p^d(p^d-1)p^{2d}-(p^d+1)=p^{3d}-p^d$.
\end{itemize}
Anyway,
$M_3=(p^{n+d}+p^n-p^d)+(p^n-1)(p^{3d}-p^d)=p^{n+3d}+p^n-p^{3d}$.
 $\square$

\begin{remark}
For the case $d'=d$, $\sum\limits_{\al,\be\in \bF_q}T(\al,\be)^3$
can also be determined, but we do not need this result.
\end{remark}

 {\it\textbf{Proof of Lemma \ref{Artin}}}:
We get that
\[
\begin{array}{rcl}
qN&=&\sum\limits_{\om\in \bF_q}\sum\limits_{x,y\in
\bF_q}\zeta_p^{\Tra_1^n\left(\om\left(\al x^{p^{3k}+1}+\be x^{p^k+1}-y^{p^d}+y\right)\right)}\\[2mm]
&=&q^2+\sum\limits_{\om\in \bF_q^*}\sum\limits_{x\in
\bF_q}\zeta_p^{\Tra_1^n\left(\om\left(\al x^{p^{3k}+1}+\be
x^{p^k+1}\right)\right)} \sum\limits_{y\in
\bF_q}\zeta_p^{\Tra_1^n\left(y^{p^d}\left(\om^{p^d}-\om\right)\right)}\\[2mm]
&=&q^2+q\sum\limits_{\om\in \bF_{q_0}^*}\sum\limits_{x\in
\bF_q}\zeta_p^{\Tra_1^n\left(\om\left(\al x^{p^{3k}+1}+\be
x^{p^k+1}\right)\right)}\\[2mm]
&=&q^2+q\sum\limits_{\om\in \bF_{q_0}^*}\sum\limits_{x\in
\bF_q}T(\om\al,\om\be)
\end{array}
\]
where the 3-rd equality follows from that the inner sum is zero
unless $\om^{p^d}-\om=0$, i.e. $\om\in \bF_{q_0}$.

For any $\om\in \bF_{q_0}^*$,  by (\ref{def H_al be}) we have
$F_{\om\al,\om\be}(X)=\om\cdot F_{\al,\be}(X)$,
$H_{\om\al,\om\be}=\om\cdot H_{\al,\be}$ and
$r_{\om\al,\om\be}=r_{\al,\be}$. From Lemma \ref{qua} (i) we know
that
\begin{equation}\label{om rel}
T({\om\al,\om\be})=\sum\limits_{X\in
\bF_{q_0}^s}\zeta_p^{\Tra_1^{d}(XH_{\om\al,\om\be}X^T)}=\eta_0(\om)^{r_{\al,\be}}T(\al,\be).
\end{equation}

In the case $d'=2d$, by Lemma \ref{rank}ii) we get that
$r_{\al,\be}$ is even. Hence $T(\om\al,\om\be)=T(\al,\be)$ for any
$\om\in \bF_{q_0}^*$ and $N=q+(p^d-1)T(\al,\be)$. $\square$

\section{Appendix B}

{\it\textbf{Proof of Theorem \ref{value dis T}}} (ii):

 In the case
$d'=2d$ ($n/d$ is even and $k/d$ is odd) and $r_{\al,\be}=s,s-2,s-4$
or $s-6$ for $(\al,\be)\neq (0,0)$. According to Lemma \ref{qua} and
Lemma \ref{reduce num}, we get that for $(\al,\be)\in N_i$,
$T(\al,\be)=(-1)^{m/d+\frac{i}{2}}p^{m+\frac{i}{2}d}$.

Combining Lemma \ref{rank} and Lemma \ref{moment} we have
\begin{equation}\label{par sum02}
 n_{0}+n_{2}+n_{4}+n_{6}=p^{2n}-1
 \end{equation}

\begin{equation}\label{par sum12}
 n_{0}-p^d\cdot n_{2}+p^{2d}\cdot n_{4}-p^{3d}\cdot n_6=(-1)^{m/d} p^m(p^{n}-1)
 \end{equation}

 \begin{equation}\label{par sum22}
 n_{0}+p^{2d}\cdot n_{2}+p^{4d}\cdot n_{4}+p^{6d}\cdot n_{6}=p^n(p^d+1)(p^n-1).
 \end{equation}

  \begin{equation}\label{par sum32}
 n_{0}-p^{3d}\cdot n_{2}+p^{6d}\cdot n_{4}-p^{9d}\cdot n_{6}=(-1)^{m/d} p^{m+3d}(p^{n}-1).
 \end{equation}
Solving the system of equations consisting of (\ref{par
sum02})--(\ref{par sum32}) yields the result. $\square$

{\it\textbf{Proof of Theorem \ref{wei dis C1}}}:
 From (\ref{Wei}) we know that for each non-zero codeword
$c(\al,\be)=\left(c_0,\cdots,c_{l-1}\right)$ $(l=q-1,
c_i=\Tra_1^n(\al\pi^{(p^{3k}+1)i}+\be \pi^{(p^k+1)i}), 0\leq i\leq
l-1, \text{and}\; (\al,\be)\in \bF_{q}\times\bF_q)$, the Hamming
weight of $c(\al,\be)$ is
\begin{equation}\label{wei c}
w_H\left(c(\al,\be)\right)=p^{n-t}(p^t-1)-\frac{1}{p^t}\cdot
R(\al,\be)
\end{equation}
where
\[R(\al,\be)=\sum\limits_{a\in \bF_{p^t}^*}T(a\al,a\be)=T(\al,\be)\sum\limits_{a\in \bF_{p^t}^*}\eta_0(a)^{r_{\al,\be}}\]
 by Lemma \ref{qua} (i).

 Let $\eta'$ be the quadratic (multiplicative) character on $\bF_q$.
 Then we have
\begin{enumerate}
    \item[(1).] if $d/t$ or $r_{\al,\be}$ is even, then $\sum\limits_{a\in \bF_{p^t}^*}\eta_0(a)^{r_{\al,\be}}=\sum\limits_{a\in \bF_{p^t}^*}
    1=p^t-1$ and $R(\al,\be)=(p^t-1)T(\al,\be)$.
    \item[(2).] if $d/t$  and $r_{\al,\be}$ are both odd, then $\sum\limits_{a\in \bF_{p^t}^*}\eta_0(a)^{r_{\al,\be}}=\sum\limits_{a\in \bF_{p^t}^*}
    \eta'(a)=0$ and $R(\al,\be)=0$.
\end{enumerate}

Thus the weight distribution of $\cC_1$ can be derived from Theorem
\ref{value dis T} and (\ref{wei c}) directly. For example, if  $d/t$
is odd and $d'=d$, then
\begin{itemize}
    \item[(1).] if
    $r_{\al,\be}=s$ and $T(\al,\be)=p^{m}$, then
    $w_H(c(\al,\be))=(p^t-1)(p^{n-t}-p^{m-t})$.
    \item[(2).]  if
    $r_{\al,\be}=s$ and $T(\al,\be)=-p^{m}$, then
    $w_H(c(\al,\be))=(p^t-1)(p^{n-t}+p^{m-t})$.
    \item[(3).] if $r_{\al,\be}=s-1$, then
    $w_H(c(\al,\be))=(p^t-1)p^{n-t}$.
    \item[(4).] if $r_{\al,\be}=s-2$ and $T(\al,\be)=-p^{m+d}$, then
    $w_H(c(\al,\be))=(p^t-1)(p^{n-t}+p^{m+d-t})$.
\end{itemize} $\square$

\section{Appendix C}

{\it\textbf{Proof of Lemma \ref{last lemma}}}:

Define $n(\al,\be,a)$ to be the number of $\ga\in \bF_q$ satisfying
(i) and (ii). From (\ref{def H_al be}) we know that
$XH_{\al,\be}X^T=\Tra_{d}^n(\al x^{p^{3k}+1}+\be x^{p^k+1})$.
Combining (\ref{def A_gamma}), (\ref{bil form1}) and (\ref{bil
form2}) we can get
\begin{eqnarray}\label{linear equivalent}
{\setlength\arraycolsep{2pt}
\begin{array}{lcl}
2XH_{\al,\be}+A_{\ga}=0& \Leftrightarrow&\;2XH_{\al,\be}Y^{T}+A_{\ga}Y^T=0\;\text{for all}\;Y\in \bF_{q_0}^s\\[2mm]
& \Leftrightarrow&\;\Tra_{d}^n\left(y\phi_{\al,\be}(x)\right)+\Tra_{d}^n(\ga y)=0\;\text{for all}\;y\in \bF_q \\[2mm]
& \Leftrightarrow&\;\Tra_{d}^n\left(y(\phi_{\al,\be}(x)+\ga)\right)=0\;\text{for all}\;y\in \bF_q \\[2mm]
&\Leftrightarrow&\phi_{\al,\be}(x)+\ga=0.
\end{array}
}
\end{eqnarray}

Let $x_0$, $x_0'$ be two distinct solutions of (i) (if exists). We
can get $x_0=X_0 \cdot V^T$ and $x'_0=X'_0\cdot V^T$ with $X_0,
X'_0\in \bF_{q_0}^s$ and $V=(v_1,\cdots, v_n)$. Define $\Delta
X_0=X'_0-X_0$ and $\Delta x_0=x'_0-x_0=X_0\cdot V^T$. Then
\[\phi_{\al,\be}(x_0)+\ga=\phi_{\al,\be}(x'_0)+\ga=0\]
gives us
\[2X_0H_{\al,\be}+A_{\ga}=2X'_0H_{\al,\be}+A_{\ga}=0\]
and hence
\[\Delta X_0\cdot H_{\al,\be}=0.\]
It follows that
\[
\begin{array}{ll}
 &X'_0\cdot  H_{\al,\be}\cdot {X'_0}^T=(X_0+\Delta X_0)\cdot H_{\al,\be}\cdot (X_0+\Delta X_0)^T\\[2mm]
 &\qquad=X_0  H_{\al,\be} {X}_0^T+ \Delta X_0\cdot H_{\al,\be}\cdot (\Delta X_0+2 X_0)=X_0  H_{\al,\be} {X}_0^T.
 \end{array}
 \]

 Therefore
\[
\begin{array}{ll}
&\Tra_{t}^m(\al {x'_0}^{p^m+1})+\Tra_{t}^n(\be
{x'_0}^{p^k+1})=\Tra_{t}^{d}\left(\Tra_{t}^m(\al
{x'_0}^{p^m+1})+\Tra_{t}^n(\be
{x'_0}^{p^k+1})\right)=\Tra_{t}^{d}\left(X'_0\cdot
H_{\al,\be}\cdot {X'_0}^T\right)\\[2mm]
&=\Tra_{t}^{d}\left(X_0 H_{\al,\be}
{X_0}^T\right)=\Tra_{t}^{d}\left(\Tra_{t}^m(\al
{x_0}^{p^m+1})+\Tra_{t}^n(\be {x_0}^{p^k+1})\right)=\Tra_{t}^m(\al
{x_0}^{p^m+1})+\Tra_{t}^n(\be {x_0}^{p^k+1}).
\end{array}
\]
Hence $n(\al,\be,a)$ is well-defined (independent of the choice of
$x_0$).

If (i) is satisfied,  that is, $\phi_{\al,\be}(x)+\ga=0$ has
solution(s) in $\bF_q$ which yields that $2XH_{\al,\be}+A_{\ga}=0$
has solution(s). Note that $\mathrm{rank}\, H_{\al,\be}=s-i$.
Therefore $2XH_{\al,\be}+A_{\ga}=0$ has $q_0^i=p^{id}$ solutions
with $X\in \bF_{q_0}^s$ which is equivalent to saying
$\phi_{\al,\be}(x)+\ga=0$ has $p^{id}$ solutions in $\bF_q$.
Conversely, for any $x_0\in \bF_q$, we can determine $\ga$ by
$\ga=-\phi_{\al,\be}(x_0).$ Let $N(\al,\be,a)$ be the number of
$x_0\in \bF_q$ satisfying (ii). Then we have
$n(\al,\be,a)=N(\al,\be,a)\big{/}p^{id}$.

Let $\chi'(x)=\zeta_p^{\Tra_1^{t}(x)}$ with $x\in \bF_{p^t}$ be an
additive character on $\bF_{p^t}$ and
$G(\eta',\chi')=\sum\limits_{x\in \bF_{p^t}}\eta'(x)\chi'(x)$ be the
Gaussian sum on $\bF_{p^t}$. We can calculate
\[
\begin{array}{rcl}
p^t\cdot N(\al,\be,a)&=&\sum\limits_{x\in \bF_q}\sum\limits_{\om\in
\bF_{p^t}}\zeta_p^{\Tra_1^{t}\left(\om\cdot\left(\Tra_{t}^n(\al
x^{p^{3k}+1}+\be
x^{p^k+1})-a\right)\right)}\\[3mm]
&=&p^n+\sum\limits_{\om \in \bF_{p^t}^*}T(\om \al,\om \be)\zeta_p^{-\Tra_1^{t}(a\om)}\\[2mm]
&=&p^n+T(\al,\be)\cdot\sum\limits_{\om \in
\bF_{p^t}^*}\eta_0(\om)^{s-i}\chi'(-a\om)
\end{array}
\]
where the 3-rd equality holds from (\ref{om rel}) for any $\om \in
\bF_{p^t}^*\subset \bF_{q_0}^*$.
\begin{itemize}
    \item If $s-i$ and $d/t$ are both odd, and $a=0$, then
    $\eta_0(\om)^{s-i}=\eta'(\om)$ and $N(\al,\be,0)=p^{n-t}$.
    \item If $s-i$ and $d/t$ are both odd, and $a\neq 0$, then
     \[
\begin{array}{rcl}
N(\al,\be,a)&=&p^{n-t}+\frac{1}{p^t}\cdot
T(\al,\be)\cdot\sum\limits_{\om \in
\bF_{p^t}^*}\eta_0(\om)\chi'(-a\om)\\[2mm]
&=&p^{n-t}+\frac{1}{p^t}\cdot T(\al,\be)\cdot\eta'(-a)\cdot G(\eta',\chi')\\[2mm]
&=&p^{n-t}+\veps \eta'(a)p^{\frac{n+id-t}{2}}
\end{array}
\]
where the 2-nd equality follows from the explicit evaluation of
quadratic Gaussian sums (see \cite{Lid Nie}, Theorem 5.15 and 5.33).
    \item If $s-i$ or $d/t$ is even, and $a=0$, then $\eta_0(\om)^{s-i}=1$
    for any $\om\in \bF_{p^t}^*$ and $N(\al,\be,0)=p^{n-t}+
    \veps (p^t-1)p^{\frac{n+id}{2}-t}$.
    \item If $s-i$ or $d/t$ is even, and $a\neq 0$, then
    \[
\begin{array}{rcl}
N(\al,\be,a)&=&p^{n-t}+\frac{1}{p^t}\cdot T(\al,\be)\cdot \sum\limits_{\om\in \bF_{p^t}^*}\chi'(-a\om)\\[2mm]
&=&p^{n-t}-\veps p^{\frac{n+id}{2}-t}.
\end{array}
\]
\end{itemize}
Therefore we complete the proof by dividing $p^{id}$. $\square$

{\it\textbf{Proof of Theorem \ref{value dis S}}} (iii): Define
\[\Xi=\left\{(\al,\be,\ga)\in \bF_q^3\left|S(\al,\be,\ga)=0 \right.\right\}\]
and $\xi=\big{|}\Xi\big{|}$.

 Recall $n_i,H_{\al,\be},,r_{\al,\be},A_{\ga}$ in Section 1 and
$N_{i,\veps},n_{i,\veps,}$ in the proof of Lemma \ref{rank}. Note
that $2XH_{0,0}+A_{\ga}=0$ is solvable if and only if $\ga=0$. If
$(\al,\be)\in N_{i,\veps}$, then the number of $\ga\in \bF_q$ such
that $2XH_{\al,\be}+A_{\ga}=0$ is solvable is $q_0^{s-i}=p^{n-id}$.

  In the case $d'=2d$, from
Lemma \ref{rank} (i) we get that
\begin{equation}\label{value xi2}
{\setlength\arraycolsep{2pt}
\begin{array}{lll}
\xi&=p^{n}-1+(p^{n}-p^{n-2d})n_{2,1}+(p^{n}-p^{n-4d})n_{4,-1}&\\[2mm]
&\begin{array}{ll} =(p^n-1)\left[1+(p^{2n}+p^{2n-9d}-\veps
p^{3m}+\veps
p^{3m-d}+\veps p^{3m-3d}\right.&\\[1mm]
\quad\left.-\veps p^{3m-5d}- \veps p^{3m-7d}+\veps
p^{3m-8d}+p^n-p^{n-d}-p^{n-4d}-p^{n-6d})/(p^d+1)\right]&
\end{array}
\end{array}
}
\end{equation}

Assume $(\al,\be)\in N_{i,\veps}$ and $\phi_{\al,\be}(x)+\ga=0$ has
solution(s) in $\bF_q$ (choose one, say $x_0$). Then by Lemma
\ref{qua} we get
\[S(\al,\be,\ga)=\zeta_p^{-\Tra_1^n\left(\alpha
x_0^{p^{3k}+1}+\beta x_0^{p^k+1}\right)}\cdot T(\al,\be).\]

Applying Lemma \ref{last lemma} for $t=1$ and Theorem \ref{value dis
T}, we get the result. $\square$

{\it\textbf{Proof of Theorem \ref{cor dis}}}: Recall
$M_{(\al_1,\be_1),(\al_2,\be_2)}(\tau)$ defined in (\ref{cor fun})
and (\ref{coe cor}).
 Fix $(\al_2,\be_2)\in
\bF_{q}\times \bF_q$, when $(\al_1,\be_1)$ runs through
$\bF_{q}\times \bF_q$ and $\tau$ takes value from $0$ to $q-2$,
$(\al',\be',\ga')$ runs through $\bF_{q}\times
\bF_q\times\left\{\bF_{q}\big{\backslash}\{1\}\right\}$ exactly one
time.

For any possible value $\kappa$ of $S(\al,\be,\ga)$, define

\[s_{\kappa}=\#\left\{(\al,\be,\ga)\in \bF_{q}\times \bF_q\times \bF_q\,\displaystyle{|}\,S(\al,\be,\ga)=\kappa\right\}\]

\[s^1_{\kappa}=\#\left\{(\al,\be,\ga)\in \bF_{q}\times \bF_q\times \left\{\bF_q\backslash\{1\}\right\}\,\big{|}\,S(\al,\be,\ga)=\kappa\right\}\]

and
\[t_{\kappa}=\#\left\{(\al,\be)\in \bF_{q}\times \bF_q\,\displaystyle{|}\,T(\al,\be)=\kappa\right\}.\]

By Lemma \ref{q-2} we have
\[s_{\kappa}^1=\frac{q-2}{q-1}\times (s_{\kappa}-t_{\kappa})+t_{\kappa}=\frac{q-2}{q-1}\times s_{\kappa}+\frac{1}{q-1}\times t_{\kappa}.\]

Define $M_{\kappa}$ to be the number of $(\al_1,\be_1,\al_2,\be_2)$
such that $M_{(\al_1,\be_1),(\al_2,\be_2)}=\kappa$. Hence we get
\[M_{\kappa}=p^{2n}\cdot s_{\kappa}^1=p^{2n}\cdot\left(\frac{q-2}{q-1}\cdot s_{\kappa}+\frac{1}{q-1}\cdot t_{\kappa}\right).\]

Then the result follows from Theorem \ref{value dis T} and Theorem
\ref{value dis S}. $\square$

{\it\textbf{Proof of Theorem \ref{wei dis C2}}}: From (\ref{Wei}) we
know that for each non-zero codeword
$c(\al,\be,\ga)=\left(c_0,\cdots,c_{q-2}\right)$ $(
c_i=\Tra_{t}^n(\al\pi^{(p^{3k}+1)i}+\be \pi^{(p^k+1)i}+\ga
\pi^{i}),\, 0\leq i\leq q-2,\, \text{and}\; (\al,\be,\ga)\in
\bF_q\times \bF_q^2)$, the Hamming weight of $c(\al,\be,\ga)$ is
\begin{equation}\label{wei c2}
w_H\left(c(\al,\be,\ga)\right)=p^{n-t}(p^t-1)-\frac{1}{p^t}\cdot
R(\al,\be,\ga)
\end{equation}
where
\[R(\al,\be,\ga)=\sum\limits_{\om\in \bF_{p^t}^*}S(\om\al,\om\be,\om\ga).\]

For any $\om\in \bF_{p^t}^*\subset \bF_{q_0}^*$, we have
$\phi_{\om\al,a\be}(x)+\om\ga=0$ is equivalent to
$\phi_{\al,\be}(x)+\ga=0$. Let $x_0\in \bF_q$ be a solution of
$\phi_{\al,\be}(x)+\ga=0$ (if exist).
\begin{itemize}
    \item[(1).]
If $\phi_{\al,\be}(x)+\ga=0$ has solutions in $\bF_q$,  then by
Lemma \ref{qua} and (\ref{om rel}) we have
\[
\begin{array}{ll}
&S(\om\al,\om\be,\om\ga)=\zeta_p^{-\left(\Tra_1^n(\om\al
x_0^{p^{3k}+1}+\om\be
x_0^{p^k+1})\right)}T(\om\al,\om\be)\\[1mm]
&\qquad=\zeta_p^{-\left(\Tra_1^n(\om\al x_0^{p^{3k}+1}+\om\be
x_0^{p^k+1})\right)}\eta_0(\om)^{r_{\al,\be}}T(\al,\be).
\end{array}
\]
 Hence
\begin{equation}\label{value of R}
R(\al,\be,\ga)=T(\al,\be)\sum\limits_{\om\in
\bF_{p^t}^*}\zeta_p^{-\Tra_1^t\left(\om\cdot\left(\Tra_t^m(\al
x_0^{p^m+1})+\Tra_t^n(\be
x_0^{p^k+1})\right)\right)}\eta_0(\om)^{r_{\al,\be}}.\nonumber
\end{equation}
    Fix $(\al,\be)\in N_{i,\veps}$ for $\veps=\pm 1$,
and suppose $\phi_{\al,\be}(x)+\ga=0$ is solvable in $\bF_q$. Denote
by $\vartheta=\Tra_t^n(\al x_0^{p^{3k}+1}+\be x_0^{p^k+1})$.
Then
\begin{itemize}
    \item if $s-i$ and $d/t$ are both odd, and $\vartheta=0$, then
\[R(\al,\be,\ga)=T(\al,\be)\sum\limits_{\om\in
\bF_{p^t}^*}\eta'(\om)=0.\]
    \item if $s-i$ and $d/t$ are both odd, and $\vartheta\neq 0$, then by the result of
quadratic Gaussian sums
\[
\begin{array}{rcl}
R(\al,\be,\ga)&=&T(\al,\be)\eta'(-\vartheta)G(\eta',\chi')\\[2mm]
&=&\veps\eta'(\vartheta)p^{\frac{n+id+t}{2}},\\[2mm]
&=&\left\{
\begin{array}{ll}
p^{\frac{n+id+t}{2}}&\text{if}\; \veps=\eta'(\vartheta),\\[2mm]
-p^{\frac{n+id+t}{2}}&\text{if}\; \veps=-\eta'(\vartheta).
\end{array}
\right.
\end{array}
\]

    \item if $s-i$ or $d/t$ is even, and $\vartheta=0$, then $\eta_0(\om)^{r_{\al,\be}}=1$ for $\om\in \bF_{p^t}^*$ and
$R(\al,\be,\ga)=(p^t-1)T(\al,\be)=\veps(p^t-1)p^{\frac{n+id}{2}}$.
    \item if $s-i$ or $d/t$ is even, and $\vartheta\neq 0$, then
    $\eta_0(\om)^{r_{\al,\be}}=1$ for $\om\in \bF_{p^t}^*$
and $R(\al,\be,\ga)=-T(\al,\be)=-\veps p^{\frac{n+id}{2}}$.
\end{itemize}
    \item[(2).] If $\phi_{\al,\be}(x)+\ga=0$ has no solutions in
    $\bF_q$ which implies that
$\phi_{\om\al,\om\be}(x)+\om\ga=0$ also has no solutions in $\bF_q$
for any $\om\in \bF_{p^t}^*\subset \bF_{q_0}$. Hence
$S(\om\al,\om\be,\om\ga)=0$ and $R(\al,\be,\ga)=0$.
\end{itemize}

 Thus the weight distribution of $\cC_2$ can be derived from Theorem \ref{value dis T},  Lemma
 \ref{last lemma}, (\ref{value xi2})
 and (\ref{wei c2}) directly.
$\square$

\section{Conclusion and Further Study}

\quad In this paper we have studied the exponential sums
$T(\al,\be)$ and $S(\al,\be,\ga)$ with $\al,\be,\ga\in \bF_q$.
After giving the value distribution of $T(\al,\be)$ and $S(\al,\be,\ga)$, we determine the
correlation distribution among a family of sequences, and the weight
distributions of the cyclic codes $\cC_1$ and $\cC_2$.

For a monomial Dembowski-Ostrom function, the associated exponential sums have been explicitly determined in
\cite{Coul1}, \cite{Coul2}. For a general Dembowski-Ostrom function $f(x)$, Lemma 1 reveals the fact that if the number of the
 solution of the linearized
polynomial related to $f(x)$ is explicitly calculated, then the exponential sums $\sum\limits_{x\in \bF_q}\chi(f(x))$ and
$\sum\limits_{x\in \bF_q}\chi(f(x)+\ga x)$ could be evaluated explicitly up to $\pm 1$. Thereafter, the correlation distribution of sequences and the weight
distributions of the associated cyclic codes are also be determined.

In particular, for the case $f(x)=x^{p^{lk}+1}+x^{p^k+1}$ with $l\geq 5$ odd, we could get the possible values of
$\sum\limits_{x\in \bF_q}\chi(f(x))$ and $\sum\limits_{x\in \bF_q}\chi(f(x)+\ga x)$.
But the first three moment identities developed in Lemma \ref{moment}
is not enough to determine the value distribution. However, we could get the possible weights of the corresponding
cyclic codes. New machinery and technique should be invented to attack this problem.

For $p=2$, the exponential sums $T(\al,\be)$ with $n/d$ odd is well known as Kasami-Welch case. Comparing to the odd characteristic
case, the binary case has one advantage since the values of $T(\al,\be)$ and $S(\al,\be,\ga)$ are all integers and one disadvantage
since the binary quadratic form theory is a little harder to handle. We will deal with the binary version of this manuscript in a
following paper.
\section{Acknowledgements}
\quad The authors will thank the anonymous referees for their
helpful comments.

\end{document}